\documentclass[a4paper,english,cleveref, autoref, thm-restate,numberwithinsect]{lipics-v2019}

\usepackage{microtype}
\usepackage{amsthm,amsmath}
\usepackage{thmtools}
\usepackage{amssymb}
\usepackage{hyperref}
\usepackage{graphicx}

\usepackage{xspace,cite}

\usepackage[utf8]{inputenc}

\usepackage{xcolor}

\usepackage[ruled, vlined,boxed,commentsnumbered,linesnumbered]{algorithm2e}
\definecolor{dark-red}{rgb}{0.4,0.15,0.15}
\definecolor{dark-blue}{rgb}{0.15,0.15,0.4}
\definecolor{medium-blue}{rgb}{0,0,0.5}
\definecolor{gray}{rgb}{0.5,0.5,0.5}
\definecolor{color-Ig}{rgb}{0.15,0.7,0.15}

\hypersetup{
    colorlinks, linkcolor={dark-red},
    citecolor={dark-blue}, urlcolor={medium-blue}
}

\newcommand{\containment}{\ensuremath{\mathsf{NP  \subseteq coNP/poly}}\xspace}
\newcommand{\notcontainment}{\ensuremath{\mathsf{NP \not \subseteq coNP/poly}}\xspace}

\newcommand{\vc}{\ensuremath{\mathrm{\textsc{vc}}}}
\newcommand{\fvs}{\ensuremath{\mathrm{\textsc{fvs}}}}

\newcommand{\cc}{\textsc{\#cc}}
\newcommand{\bdmod}{{\sf bd\text{-}mod}\xspace}

\newcommand{\twmod}{{\sf tw\text{-}mod}\xspace}

\newcommand{\distto}{{\sf dist\text{-}to\text{-}}}
\newcommand{\fmod}{{\sf f\text{-}mod}}
\newcommand{\mc}{\mathcal}
\newcommand{\bd}{{\sf bd}}
\newcommand{\td}{{\sf td}}
\newcommand{\tw}{{\sf tw}}

\newcommand{\nm}{{\sf nm}}
\newcommand{\tpm}{{\sf tpm}}
\newcommand{\diam}{{\sf diam}}
\newcommand{\mbs}{{\sf mbs}}
\newcommand{\X}{\mathcal{X}}
\newcommand{\conf}{{\sf conf}\xspace}
\renewcommand{\bar}[1]{{#1}_{\textsc{cb}}}
\newcommand{\barp}[1]{{#1}'_{\textsc{cb}}}
\renewcommand{\O}{\mathcal{O}}
\renewcommand{\P}{\mathcal{P}}
\newcommand{\F}{\mathcal{F}}

\newcommand{\IS}{\mbox{\sc IS}\xspace}
\newcommand{\VC}{\textsc{VC}\xspace}

\newcommand{\pb}{\IS/$c$-\bdmod}

\newcommand{\todo}[1][]{%
  \ifx/#1/%
    \textcolor{red}{TODO!}%
  \else%
    \textcolor{red}{todo: #1}%
  \fi%
}
\newcommand{\Oh}{\ensuremath{\mathcal{O}}\xspace}

\newcommand{\defparproblem}[4]{\par
 \vspace{3mm}
\noindent\fbox{
 \begin{minipage}{0.96\textwidth}
 \begin{tabular*}{\textwidth}{@{\extracolsep{\fill}}lr} #1 & {\bf{Parameter:}} #3 \vspace{1mm} \\ \end{tabular*}
 {\textbf{Input:}} #2
	\vspace{1mm}\\%
 {\textbf{Question:}} #4
 \end{minipage}
 }
 \vspace{3mm}
\par
}

\theoremstyle{plain}

\newtheorem{observation}[theorem]{Observation}

\newtheorem{ruleN}{Rule}
\newtheorem{MetaruleN}{Meta-Rule}

\newenvironment{ruleproof}[1][Proof of safeness]{\begin{proof}[#1]}{\end{proof}}

\title{Bridge-Depth Characterizes which Minor-Closed Structural Parameterizations of Vertex Cover Admit a Polynomial Kernel}

\titlerunning{Bridge-Depth Characterizes the Kernelization Complexity of Vertex Cover} 

\author{Marin Bougeret}{LIRMM, Universit\'e de Montpellier, France}{marin.bougeret@lirmm.fr}{https://orcid.org/0000-0002-9910-4656}{}

\author{Bart M.\ P.\ Jansen}{Eindhoven University of Technology, The Netherlands}{b.m.p.jansen@tue.nl}{https://orcid.org/0000-0001-8204-1268}{
Supported by the European Research Council (ERC) under the European Union's Horizon 2020 research and innovation programme (grant agreement No 803421, ReduceSearch).\\ \includegraphics[height=2cm]{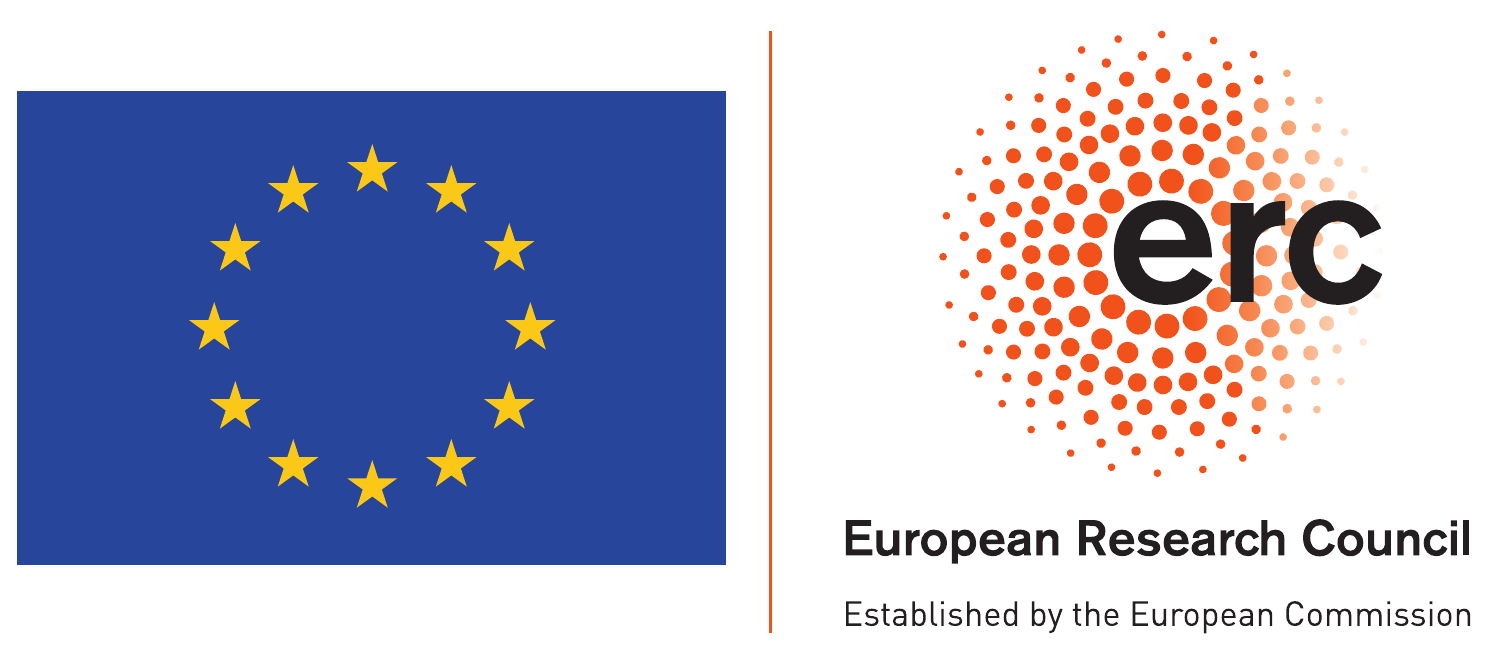}}

\author{Ignasi Sau}{LIRMM, Universit\'e de Montpellier, CNRS, France}{ignasi.sau@lirmm.fr}{https://orcid.org/0000-0002-8981-9287}{Supported  by French projects DEMOGRAPH (ANR-16-CE40-0028), ESIGMA (ANR-17-CE23-0010), and ELIT (ANR-20-CE48-0008-01).}

\authorrunning{M. Bougeret and B. M. P. Jansen and I. Sau} 

\Copyright{M. Bougeret and B. M. P. Jansen and I. Sau} 

\ccsdesc[100]{Theory of computation~Graph algorithms analysis}
\ccsdesc[100]{Theory of computation~Parameterized complexity and exact algorithms}

\keywords{vertex cover, parameterized complexity, polynomial kernel, structural parameterization, bridge-depth.}

\category{}

\relatedversion{A conference version of this paper appeared in the \emph{Proceedings of the 47th International Colloquium on Automata, Languages and Programming (\textbf{ICALP}), volume 168 of LIPIcs, pages 16:1-16:19, July \textbf{2020}}.} 

\supplement{}

\acknowledgements{}

\nolinenumbers 

\hideLIPIcs  

\EventEditors{Artur Czumaj, Anuj Dawar, and Emanuela Merelli}
\EventNoEds{3}
\EventLongTitle{47th International Colloquium on Automata, Languages and Programming (ICALP 2020)}
\EventShortTitle{ICALP 2020}
\EventAcronym{ICALP}
\EventYear{2020}
\EventDate{July 8--11, 2020}
\EventLocation{Saarbrücken, Germany}
\EventLogo{}
\SeriesVolume{168}
\ArticleNo{16}

\begin{document}

\maketitle

\begin{abstract}
  We study the kernelization complexity of structural parameterizations of the \textsc{Vertex Cover} problem. Here, the goal is to find a polynomial-time preprocessing algorithm that can reduce any instance~$(G,k)$ of the \textsc{Vertex Cover} problem to an equivalent one, whose size is polynomial in the size of a pre-determined complexity parameter of~$G$. A long line of previous research deals with parameterizations based on the number of vertex deletions needed to reduce~$G$ to a member of a simple graph class~$\mathcal{F}$, such as forests, graphs of bounded tree-depth, and graphs of maximum degree two. We set out to find the most general graph classes~$\mathcal{F}$ for which \textsc{Vertex Cover} parameterized by the vertex-deletion distance of the input graph to~$\mathcal{F}$, admits a polynomial kernelization. We give a complete characterization of the minor-closed graph families~$\mathcal{F}$ for which such a kernelization exists. We introduce a new graph parameter called \emph{bridge-depth}, and prove that a polynomial kernelization exists if and only if~$\mathcal{F}$ has bounded bridge-depth. The proof is based on an interesting connection between bridge-depth and the size of minimal blocking sets in graphs, which are vertex sets whose removal decreases the independence number. \bigskip
\end{abstract}

\newpage

\section{Introduction} \label{sec:intro}
\subparagraph*{Background and motivation.}

The NP-complete \textsc{Vertex Cover} problem is one of the most prominent problems in the field of kernelization~\cite{Bodlaender16,CyganFKLMPPS15,FellowsJKRW18,FominLSZ19,LokshtanovMS12}, which investigates provably efficient and effective preprocessing for parameterized problems. A \emph{parameterized problem} is a decision problem in which a positive integer~$k$, called the \emph{parameter}, is associated with every instance~$x$. A \emph{kernelization} for a parameterized problem is a polynomial-time algorithm that reduces any parameterized instance~$(x,k)$ to an equivalent instance~$(x',k')$ of the same problem whose size is bounded by~$f(k)$ for some function~$f$, which is the \emph{size} of the kernelization. Hence a kernelization guarantees that instances that are large compared to their parameter, can be efficiently reduced without changing their answer. Of particular interest are \emph{polynomial kernelizations}, whose size bound~$f$ is polynomial.

An instance~$(G,k)$ of \textsc{Vertex Cover} asks whether the undirected graph~$G$ has a vertex set~$S$ of size at most~$k$ that contains at least one endpoint of every edge. Using the classic Nemhauser-Trotter theorem~\cite{NemhauserT75}, one can reduce~$(G,k)$ in polynomial time to an instance~$(G',k')$ with the same answer, such that~$|V(G')| \leq 2k$. Hence when using the size of the desired solution as the parameter, \textsc{Vertex Cover} has a kernelization that reduces to instances of~$2k$ vertices, which can be encoded in~$\Oh(k^2)$ bits. While the bitsize of this kernelization is known to be essentially optimal~\cite{DellM14} assuming the established conjecture~\notcontainment, this result does not guarantee any effect of the preprocessing for instances whose solution has size at least~$|V(G)| / 2$. In particular, it does not promise any size reduction when~$G$ is simply a path.

To be able to give better preprocessing guarantees, one can use \emph{structural parameters} that take on smaller values than the size of a minimum vertex cover, a quantity henceforth called the \emph{vertex cover number}. Such structural parameterizations can conveniently be described in terms of the vertex-deletion distance to certain graph families~$\mathcal{F}$. Note that the vertex cover number~$\vc(G)$ of~$G$ can be defined as the minimum number of vertex deletions needed to reduce~$G$ to an edgeless graph. Hence this number will always be at least as large as the \emph{feedback vertex number}~$\fvs(G)$ of~$G$, which is the vertex-deletion distance of~$G$ to a forest. In 2011, it was shown that \textsc{Vertex Cover} even admits a polynomial kernelization when parameterized by the feedback vertex number~\cite{JansenB11,JansenB13}. This triggered a long line of follow-up research, which aimed to find the most general graph families~$\mathcal{F}$ such that \textsc{Vertex Cover} admits a polynomial kernelization when parameterized by vertex-deletion distance to~$\mathcal{F}$. Polynomial kernelizations were obtained for the families~$\mathcal{F}$ of graphs of maximum degree two~\cite{MajumdarRR18}, of graphs of constant \emph{tree-depth}~\cite{BougeretS18,JansenP18}, of the \emph{pseudo-forests} where each connected component has at most one cycle~\cite{FominS16}, and for $d$-\emph{pseudo-forests} in which each connected component has a feedback vertex set of size at most~$d \in \Oh(1)$~\cite{HolsK17}. Note that all these target graph classes are \emph{closed under taking minors}. Using \emph{randomized algorithms} with a small error probability, polynomial kernelizations are also known for several parameterizations by vertex-deletion distance to graph classes that are not minor-closed, such as K\H{o}nig graphs~\cite{KratschW12}, bipartite graphs~\cite{KratschW12}, and parameterizations based on the linear-programming relaxation of \textsc{Vertex Cover}~\cite{HolsKP19,Kratsch18}. On the negative side, it is known that \textsc{Vertex Cover} parameterized by the vertex-deletion distance to a graph of treewidth two~\cite{CyganLPPS14} does \emph{not} have a polynomial kernel, unless \containment. This long line of research into kernelization for structural parameterizations raises the following question:

\begin{quote}
How can we characterize the graph families~$\mathcal{F}$ for which \textsc{Vertex Cover} parameterized by vertex-deletion distance to~$\mathcal{F}$ admits a polynomial kernel?
\end{quote}

\noindent We answer this question for all minor-closed families~$\mathcal{F}$, by introducing a new graph parameter.

\subparagraph*{Our results.}

We introduce a new graph parameter that we call \emph{bridge-depth}. It has a recursive definition similar to that of tree-depth~\cite{NesetrilM06} (full definitions follow in Section~\ref{sec:introducing}), but deals with bridges in a special way. A graph without vertices has bridge-depth zero. The bridge-depth~$\bd(G)$ of a disconnected graph~$G$ is simply the maximum bridge-depth of its connected components. The bridge-depth of a connected nonempty graph~$G$ is defined as follows. Let~$\bar{G}$ denote the graph obtained from~$G$ by contracting each edge that is a bridge in~$G$; the order does not matter. Then~$\bd(G) := 1 + \min_{v \in V(\bar{G})} \bd(\bar{G} \setminus v)$. Intuitively, the bridge-depth of~$G$ is given by the depth of an elimination process~\cite{BulianD16} that reduces~$G$ to the empty graph. One step consists of contracting all bridges and removing a vertex; each of the remaining connected components is then recursively eliminated in parallel. From this definition, it is not difficult to see that~$\bd(G)$ is at least as large as the \emph{tree-width} of~$G$, but never larger than the tree-depth or feedback vertex number of~$G$. In particular, any forest has bridge-depth one.

Using the notion of bridge-depth, we characterize the minor-closed families~$\mathcal{F}$ for which \textsc{Vertex Cover} parameterized by vertex-deletion distance to~$\mathcal{F}$ admits a polynomial kernel.

\begin{theorem} \label{thm:characterization}
Let~$\mathcal{F}$ be a minor-closed family of graphs, and assume \notcontainment. \textsc{Vertex Cover} parameterized by vertex-deletion distance to~$\mathcal{F}$ has a polynomial kernelization if and only if~$\mathcal{F}$ has bounded bridge-depth.
\end{theorem}

Theorem~\ref{thm:characterization} gives a clean and unified explanation for all the minor-closed families~$\mathcal{F}$ that were previously considered individually~\cite{BougeretS18,FominS16,HolsK17,JansenB13,MajumdarRR18}, and generalizes these results as far as possible. To the best of our knowledge, Theorem~\ref{thm:characterization} captures all known (deterministic) kernelizations for structural parameterizations of \textsc{Vertex Cover}. (There are \emph{randomized} kernelizations~\cite{HolsKP19,Kratsch18,KratschW12} that apply for distance to classes~$\mathcal{F}$ that are \emph{not} minor-closed, such as bipartite graphs.) For example, we capture the case of~$\mathcal{F}$ being a forest~\cite{JansenB13} since forests have bridge-depth one, and  the case of~$\mathcal{F}$ being graphs of constant tree-depth~\cite{BougeretS18,JansenP18} since bridge-depth does not exceed tree-depth.
In this sense, bridge-depth can be seen as the ultimate common generalization of feedback vertex number and tree-depth (which are incomparable parameters) in the context of polynomial kernels for \textsc{Vertex Cover}.

We consider it one of our main contributions to identify the graph parameter bridge-depth as the right way to capture the kernelization complexity of \textsc{Vertex Cover} parameterizations.

\subparagraph*{Techniques.} To describe our techniques, we introduce some terminology. Let~$\alpha(G)$ denote the independence number of graph~$G$, i.e., the maximum size of a set of pairwise nonadjacent vertices. A \emph{blocking set} in a graph~$G$ is a vertex set~$Y \subseteq V(G)$ such that~$\alpha(G \setminus Y) < \alpha(G)$. Hence if~$Y$ is a blocking set, then every maximum independent set in~$G$ contains a vertex from~$Y$. The connection between blocking sets and the kernelization of \textsc{Vertex Cover} has been already exploited in previous work. Namely, earlier kernelizations for \textsc{Vertex Cover} parameterized by distance to a graph class~$\mathcal{F}$, starting with the work of Jansen and Bodlaender~\cite{JansenB13}, all rely, either implicitly or explicitly, on having upper-bounds on the size of (inclusion-)minimal blocking sets for graphs in~$\mathcal{F}$~\cite{BougeretS18,FominS16,HolsK17,JansenB13,MajumdarRR18}. For example, it is known that minimal blocking sets in a bipartite graph have size at most two~\cite[Cor.~11]{HolsK17}, while minimal blocking sets in graphs of tree-depth~$c$ have size at most~$2^c$~\cite[Lemma 1]{BougeretS18}. Similarly, all the existing superpolynomial kernelization lower bounds for parameterizations by distance to~$\mathcal{F}$, rely on~$\mathcal{F}$ having minimal blocking sets of arbitrarily large size. Indeed, if~$\mathcal{F}$ is closed under disjoint union and has arbitrarily large blocking sets, it is easy to prove a superpolynomial lower bound (cf.~\cite[Thm.~1]{HolsKP19}).

Since all positive cases for kernelization are when minimal blocking sets of graphs in~$\mathcal{F}$ have bounded size, while one easily obtains lower bounds when the size of minimal blocking sets of graphs in~$\mathcal{F}$ is unbounded, the question rises whether a bound on the size of minimal blocking sets is a necessary and sufficient condition for the existence of polynomial kernels. To our initial surprise, we show that for minor-closed families~$\mathcal{F}$, this is indeed the case: the purely structural property of having bounded-size minimal blocking sets can always be leveraged into preprocessing algorithms.

For an insight into our techniques, consider an instance~$(G,k)$ of \textsc{Vertex Cover}, together with a vertex set~$X \subseteq V(G)$ such that~$G \setminus X \in \mathcal{F}$ for some minor-closed family~$\mathcal{F}$ that has bounded-size minimal blocking sets.
The goal of the kernelization is then to reduce to an equivalent instance of size~$|X|^{\Oh(1)}$ in polynomial time. Using ideas of the previous kernelizations~\cite{JansenB13,BougeretS18}, it is quite simple to reduce the \emph{number of connected components} of~$G \setminus  X$ to~$|X|^{\Oh(1)}$. To obtain a polynomial kernel, the challenge is therefore to bound the size of each such component~$C$ of~$G \setminus  X$ to~$|X|^{\Oh(1)}$, so that the overall instance size becomes polynomial in~$|X|$. However, the non-existence of large minimal blocking sets does not seem to offer any handle for reducing the size of individual components of~$G \setminus X$. The route to the kernelization therefore goes via the detour of bridge-depth. We prove the following relation between the sizes of minimal blocking sets and bridge-depth.

\begin{theorem} \label{thm:bridgedepth:blockingsets}
Let~$\mathcal{F}$ be a minor-closed family of graphs. Then~$\mathcal{F}$ has bounded bridge-depth if and only if the size of minimal blocking sets of graphs in~$\mathcal{F}$ is bounded.
\end{theorem}

Using this equivalence, we can exploit the fact that all minimal blocking sets of~$\mathcal{F}$ are of bounded size, through the fact that the bridge-depth of~$G \setminus X \in \mathcal{F}$ is small. This means that there is a bounded-depth elimination process to reduce~$G \setminus X$ to the empty graph. We use this bounded-depth process in a technical kernelization algorithm following a recursive scheme, inspired by the earlier kernelization for the parameterization by distance to bounded tree-depth~\cite{BougeretS18}.


Let us now discuss the ideas behind the equivalence of Theorem~\ref{thm:bridgedepth:blockingsets}. We prove that the bridge-depth of graphs in a minor-closed family~$\mathcal{F}$ is upper-bounded in terms of the maximum size of minimal blocking sets for graphs in~$\mathcal{F}$, by exploiting the Erd\H{o}s-P\'osa property in an interesting way. We analyze an elementary graph structure called \emph{necklace of length~$t$}, which is essentially the multigraph formed by a path of~$t$ double-edges. If a simple graph~$G \in \mathcal{F}$ contains a necklace of length~$t$ as a minor, then there is a minor~$G'$ of~$G$ (which therefore also belongs to~$\mathcal{F}$) that has a minimal blocking set of size~$\Omega(t)$. Hence to show that bridge-depth is upper-bounded in terms of the size of minimal blocking sets of graphs in~$\mathcal{F}$, it suffices to show that bridge-depth is upper-bounded by the maximum length of a necklace minor of graphs in~$\mathcal{F}$. Since the definition of bridge-depth allows for the contraction of all bridges in a single step, it suffices to consider bridgeless graphs. Then we argue that in a bridgeless graph, any pair of maximum-length necklace minor models intersects at a vertex (cf. Lemma~\ref{lemma:pack}). By the Erd\H{o}s-P\'osa property, this implies that there is a constant-size vertex set that hits all maximum-length necklace minor models, and whose removal therefore strictly decreases the maximum length of a necklace minor. If the length of necklace minor models is bounded, then after a bounded number of steps of this process (interleaved with contracting all bridges) we reduce the maximum length of necklace minor models to zero, which is equivalent to breaking all cycles of the graph. At that point, the bridge-depth is~one by definition, and we have obtained the desired upper-bound on the bridge-depth in terms of the length of the longest necklace minor, and therefore blocking set size.

For the other direction of Theorem~\ref{thm:bridgedepth:blockingsets}, we prove (cf.~Theorem~\ref{theorem:blockingset:size}) the tight bound that a minimal blocking set in a graph~$G$ has size at most~$2^{\bd(G)}$. We use induction to prove this statement, together with an analysis of the structure  of a tree of bridges whose removal
decreases the bridge-depth. The fact that bipartite graphs have minimal blocking sets of size at most two, allows for an elegant induction step.


\subparagraph*{Related work.}
In a recent paper,  Hols, Kratsch, and Pieterse~\cite{HolsKP19} also analyze the role of blocking sets in the existence of polynomial kernels for structural parameterizations of \textsc{Vertex Cover}. Note that our paper is independent from, and orthogonal to~\cite{HolsKP19}: we consider the setting of \emph{deterministic} kernelization algorithms for parameterizations to \emph{minor-closed} families~$\mathcal{F}$, and obtain an exact characterization of which~$\mathcal{F}$ allow for a polynomial kernelization. Hols et al.~\cite{HolsKP19} consider hereditary families~$\mathcal{F}$  and give kernelizations for several such parameterizations, without arriving at a complete characterization. Some of the randomized kernelizations they provide do not fit into our framework, but all the deterministic kernelizations they present are captured by Theorem~\ref{thm:characterization}. Another contribution of~\cite{HolsKP19} is to prove that there is a class $\mathcal{F}$ with minimal blocking sets of size one where \textsc{Vertex Cover} cannot be solved in polynomial time. In particular, there is no polynomial kernel parameterized by the distance to this family $\mathcal{F}$, and thus bounded minimal blocking set size is not sufficient to get a polynomial kernel. This implies that our minor-closed assumption of Theorem~\ref{thm:characterization} cannot be dropped.

We refer to the survey by Fellows et al.~\cite{FellowsJKRW18} for an overview of classic results and new research lines concerning kernelization for \textsc{Vertex Cover}. Additional relevant work includes the work by Kratsch~\cite{Kratsch18} on a randomized polynomial kernel for a parameterization related to the difference between twice the cost of the linear-programming relaxation of \textsc{Vertex Cover} and the size of a maximum matching.

\subparagraph*{Organization.}
Preliminaries on graphs and complexity are presented in Section~\ref{sec:preliminaries}. Section~\ref{sec:introducing} introduces bridge-depth and its properties. In Section~\ref{sec:triangle-path} we prove one direction of Theorem~\ref{thm:bridgedepth:blockingsets}, showing that large bridge-depth implies the existence of large minimal blocking sets. In Section~\ref{sec:blockingsets} we handle the other direction, proving a tight upper-bound on the size of minimal blocking sets in terms of the bridge-depth. We present the kernelization algorithm exploiting bridge-depth in Section~\ref{sec:kernel:summary}, and we conclude the article in Section~\ref{sec:conclusion}.


 \section{Preliminaries} \label{sec:preliminaries}

\subparagraph*{Graphs.} We use standard graph-theoretic notation, and we refer the reader to Diestel~\cite{Diestel16} for any undefined terms. All graphs we consider are finite and undirected. Graphs are simple, unless specifically stated otherwise, such as when dealing with necklaces (see Definition~\ref{def:necklace}). A graph~$G$ has vertex set~$V(G)$ and edge set~$E(G)$.
Given a graph $G$ and a subset $S \subseteq V(G)$, we say that $S$ is \emph{connected} if $G[S]$ is connected, and we use the shorthand $G \setminus S$ to denote $G[V(G) \setminus S]$. For a single vertex~$v \in V(G)$, we use~$G \setminus v$ as a shorthand for~$G \setminus \{v\}$. Similarly, for a set of edges~$T \subseteq E(G)$ we denote by~$G \setminus T$ the graph on vertex set~$V(G)$ with edge set~$E(G) \setminus T$. A cycle on three vertices is called a \emph{triangle}. For two positive integers $i,j$ with $i \leq j$, we denote by $[i,j]$ the set of all integers $\ell$ such that $i \leq \ell \leq j$, and by $[i]$ the set $[1,i]$. Given $v \in V(G)$, we denote $N_G(v)=\{u \mid \{u,v\} \in E(G)\}$, $d_G(v)=|N_G(v)|$
and given $X \subseteq V(G)$, we denote $N_G(X)=\bigcup_{v \in X}N_G(v) \setminus X$.
Given $X,Y \subseteq V(G)$, we denote by $N_G^Y(X)=N_G(X) \cap Y$. We may omit the subscript~$G$ when it is clear from the context.
For distinct vertices $u$ and $v$ of a graph $G$, the graph $G'$ obtained by \emph{identifying} $u$ and $v$ is defined by removing vertices $u$ and $v$ from~$G$, adding a new vertex $uv$ with $N_{G'}(uv) = N_G(\{u,v\})$, and keeping the other vertices and edges unchanged. Given two adjacent vertices $u$ and $v$, we define the \emph{contraction} of the edge $\{u,v\}$ as the identification of $u$ and $v$.

Given a graph $G$, we denote by $\alpha(G)$ the size of a maximum independent set in $G$, by $\cc(G)$ the number of connected components of $G$,
by $\diam(G)$ the diameter of $G$, and by $\Delta(G)$ the maximum degree of $G$.
Given an independent set $S \subseteq V(G)$ and a set $V' \subseteq V(G)$, we denote $S_{V'} = S \cap V'$.
Given a graph $G$ and a set $S \subseteq V(G)$, we say that $S$ is a \emph{blocking set} in $G$ if $\alpha(G \setminus S) < \alpha(G)$. The maximum size of an inclusion-wise minimal blocking set of a graph $G$ is denoted by $\mbs(G)$.

A graph~$H$ is a minor of graph~$G$ if~$H$ can be obtained from~$G$ by a sequence of edge deletions, edge contractions, and removals of isolated vertices. Let us also recall the definition of minor in the context of multigraphs.
Let $H$ be a loopless multigraph. An \emph{$H$-model} $M$ in a simple graph $G$ is a collection $\{S^M_x \mid x \in V(H)\}$ of pairwise disjoint subsets of $V(G)$ such that $G[S^M_x]$ is connected for every $x \in V(H)$, and such that for every pair of distinct vertices $x,y$ of $H$, the quantity~$|\{ \{u,v\} \in E(G) \mid u \in S^M_x, v \in S^M_y \}|$ is at least the number of edges in $H$ between $x$ and $y$.
The vertex set $V(M)$ of $M$ is the union of the vertex sets of the subgraphs in the collection.
We say that a graph $G$ contains a loopless multigraph $H$ as a \emph{minor} if $G$ has an $H$-model.

For the following definitions, we refer the reader to~\cite{NeMe12} for more details and we only recall here some basic notations and facts.
The \emph{tree-depth} of a graph~$G$, denoted by~$\td(G)$, is defined recursively. The empty graph without vertices has tree-depth zero. The tree-depth of a disconnected graph is the maximum tree-depth of its connected components. Finally, if~$G$ is a nonempty connected graph then~$\td(G) = 1 + \min_{v \in V(G)} \td(G \setminus v)$. Equivalent definitions exist in terms of the minimum height of a rooted forest whose closure is a supergraph of~$G$. The tree-width of~$G$ is denoted~$\tw(G)$ (cf.~\cite{Bodlaender98}).

Given a graph family $\F$, an $\F$-\emph{modulator} in a graph $G$ is a subset of vertices $X \subseteq V(G)$ such that $G \setminus X \in \F$.
We denote by \distto$\F(G)$ the size of a smallest $\F$-modulator in $G$.
For a graph measure ${\sf f}$ that associates an integer with each graph, and an integer $c$, a $c$-\emph{${\sf f}$-modulator} is a modulator to $\F^{\sf f}_c := \{G \mid {\sf f}(G) \le c\}$.
We denote by $c$-$\fmod(G):=\distto\F^{\sf f}_c(G)$, that is, the size of a smallest $c$-${\sf f}$-modulator of $G$.
Typical measures ${\sf f}$ that we consider here are tree-width, tree-depth, and bridge-depth.
Notice that $0$-$\twmod(G)$ corresponds to the minimum size of a vertex cover of $G$, and $1$-$\twmod(G)$ corresponds to the minimum size of a feedback vertex set of $G$.
Finally, \IS (resp. \VC) denotes the \textsc{Maximum Independent Set} (resp. \textsc{Minimum Vertex Cover})  problem.

\subparagraph*{Parameterized complexity.}
A \emph{parameterized problem} is a language $L \subseteq \Sigma^* \times \mathbb{N}$, for some finite alphabet $\Sigma$.  For an instance $(x,k) \in \Sigma^* \times \mathbb{N}$, the value~$k$ is called the \emph{parameter}. For a computable function~$g \colon \mathbb{N} \to \mathbb{N}$, a \emph{kernelization algorithm} (or simply a \emph{kernel}) for a parameterized problem $L$ of \emph{size} $g$ is an algorithm $A$ that given any instance $(x,k)$ of $L$, runs in polynomial time and returns an instance $(x',k')$ such that $(x,k) \in L \Leftrightarrow (x', k') \in L$ with $|x'|, k' \le g(k)$. Consult~\cite{CyganFKLMPPS15,Niedermeier06,FlumG06,DoFe13,FominLSZ19} for background on parameterized complexity.

 \section{An introduction to bridge-depth} \label{sec:introducing}

  Let $G$ be a graph. An edge $e \in E(G)$ is a \emph{bridge} if its removal increases the number of connected components of $G$.
  We define $\bar{G}$ as the simple graph obtained from $G$ by contracting all bridges of $G$ (the order does not matter.)
Observe that, as contracting an edge cannot create a new bridge, $\bar{G}$ has no bridges, implying that $\bar{(\bar{G})}=\bar{G}$.
Given a  subgraph $T$ of a graph $G$, we say that $T$ is a \emph{tree of bridges} if $T$ is a tree and, for every $e \in E(T)$, $e$ is a bridge in $G$. Note that a single vertex is, by definition, a tree of bridges.
Note also that with any vertex $v \in V(\bar{G})$ we can associate, in a bijective way, an inclusion-wise maximal tree of bridges $T_v$ of $G$.
The set $\{T_v \mid v \in V(\bar{G})\}$ is a minor model of $\bar{G}$ in~$G$ (a $\bar{G}$-model, from now on). For any $u,v \in V(\bar{G})$ such that $\{u,v\} \in E(\bar{G})$, there is exactly one edge
$\{u',v'\} \in E(G)$ with $u' \in T_u$ and $v' \in T_v$. The latter claim can be easily verified by supposing that there are two such edges, implying that some edge in $T_u$ or $T_v$ is involved in a cycle, which contradicts the fact that all the edges in $T_u$ and $T_v$ are bridges.

\begin{definition}\label{def:bd}
  The \emph{bridge-depth}~$\bd(G)$ of a graph $G$ is recursively defined as follows:
  \begin{itemize}
  \item If $G$ is the empty graph without any vertices, then~$\bd(G)=0$.
  \item If $G$ has $\ell > 1$ connected components $\{ G_i \mid i \in [\ell]\}$, then~$\bd(G)=\max_{i \in [\ell]} \bd(G_{i})$.
  \item If $G$ is connected, then~$\bd(G)= 1+\min_{v \in V(\bar{G})} \bd(\bar{G} \setminus v)$.
  \end{itemize}
\end{definition}

Informally, $\bd$ behaves like tree-depth except that at each step of the recursive definition we are allowed to delete trees of bridges instead of just single vertices, as proved in Item~\ref{rectreeofbridges} of the following proposition. The following properties of bridge-depth follow from the definitions in an elementary way, often exploiting the fact that if~$e$ is a bridge in~$G$, then~$e$ is also a bridge in any minor of~$G$ that still contains~$e$.

\begin{restatable}{proposition}{propositionXbd} \label{proposition:bd}
 For any graph $G$ the following claims hold:

  \begin{enumerate}
  \item $\bd(G)=1$ if and only if $G$ is a forest with at least one vertex.
  \item $\bd(\bar{G})=\bd(G)$.\label{contractionsstable}
  \item The parameter $\bd$ is minor-closed: if~$G'$ is a minor of~$G$ then~$\bd(G') \le \bd(G)$. \label{minorclosed}  
  \item If $G$ is connected, then~$\bd(G)= 1+\min_{T} \bd(G \setminus V(T))$, where the minimum is taken over all  trees of bridges $T$ of $G$. \label{rectreeofbridges}
  \item For any $X \subseteq V(G)$, we have~$\bd(G) \le |X|+\bd(G \setminus X)$. \label{removeX}
	\item $\tw(G) \leq \bd(G)$. \label{item:treewidth} 
  \end{enumerate}
\end{restatable}
\begin{proof}
The first item follows easily from the definition, while the second one uses that $\bar{(\bar{G})}=\bar{G}$.

\medskip
 \textbf{Proof of~\ref{minorclosed}}: We prove the claim by induction on $|V(G)| + |E(G)|$. Suppose that~$G$ has multiple connected components~$\{G_i \mid i \in [\cc(G)]\}$, and let~$\{G'_i \mid i \in [\cc(G')]\}$ be the connected components of the minor~$G'$ of~$G$. Then each connected component~$G'_j$ is a minor of some component~$G_i$ of~$G$ on fewer than~$|V(G)|$ vertices, which gives~$\bd(G'_j) \leq \bd(G_i)$ by induction. Hence we have~$\bd(G') = \max_{j \in [\cc(G')]} \bd(G'_j) \leq \max_{i \in [\cc(G)]} \bd(G_i) = \bd(G)$.

We now deal with the case that $G$ is connected. In general, if some graph~$G^*$ is a minor of~$G$, then~$G^*$ is a minor of a graph~$G'$ obtained from~$G$ by removing an edge, contracting an edge, or removing an isolated vertex. Since~$G$ is assumed to be connected, the third case cannot occur here. Then by induction, we have~$\bd(G^*) \leq \bd(G')$, so it suffices to prove that~$\bd(G') \leq \bd(G)$ for any graph~$G'$ obtained by removing or contracting an edge.
  Let us first prove that if $\barp{G}$ is a minor of $\bar{G}$, then $\bd(G') \le \bd(G)$. Indeed, let~$v^* \in V(\bar{G})$ such that~$\bd(G) = 1 + \bd(\bar{G} \setminus v^*)$, and consider an arbitrary component~$G'_i$ of~$G'$. Note that~$\bar{(G'_i)}$ is a component of~$\bar{(G')}$, and therefore a minor of~$\bar{G}$ by hypothesis. If~$\bar{(G'_i)}$ is a minor of the graph~$\bar{G} \setminus v^*$, then by induction and Item~\ref{contractionsstable} we have~$\bd(G'_i) = \bd(\bar{(G'_i)}) \leq \bd(\bar{G} \setminus v^*) < \bd(G)$. Otherwise, any minor model~$\{S_x \mid x \in V(\bar{(G'_i)})\}$ of~$\bar{(G'_i)}$ in~$\bar{G}$ contains a branch set~$S_{x^*}$ with~$v^* \in S_{x^*}$. But then~$\bd(\bar{(G'_i)}) \leq 1 + \bd(\bar{(G'_i)} \setminus x^*)$ by definition, and~$\bar{(G'_i)} \setminus x^*$ is a minor of~$\bar{G} \setminus v^*$, and therefore has bridge-depth at most~$\bd(G) - 1$, so that~$\bd(\bar{(G'_i)}) \leq \bd(G)$. Hence for each component~$G'_i$ of~$G'$ we have~$\bd(G'_i) = \bd(\bar{(G'_i)}) \leq \bd(G)$, implying~$\bd(G') \leq \bd(G)$.
	
  Thus, it only remains to prove that $\barp{G}$  is a minor of $\bar{G}$.
  Let us first assume that $G'$ is obtained from $G$ by removing an edge $e$.  Let $\{T_v \mid v \in V(\bar{G})\}$ be the $\bar{G}$-model in $G$ given by the trees of bridges.
  If $e$ is not a bridge, then $e$ is an edge between $T_u$ and $T_v$ for some vertices $u,v \in V(\bar{G})$.
  To obtain $\barp{G}$ as a minor, we start from $\bar{G}$, remove edge $\{u,v\}$, and for any edge $e'$ between $T_{u'}$ and $T_{v'}$ (for any~$u',v' \in V(\bar{G})$)
	that has become a bridge in $G'$ because of the removal of $e$, we contract $\{u',v'\}$. This implies that $\barp{G}$ is a minor of $\bar{G}$.
  Otherwise, if $e$ is a bridge, then there exists $u \in V(\bar{G})$ such that $e \in  E(T_u)$, and $G'$ has two connected components $G'_1$ and $G'_2$.
  To obtain $\bar{(G'_i)}$ as a minor, for $i \in [2]$, we start from $\bar{G}$ and remove any vertex $v$ such that $T_v \cap V(G'_i) = \emptyset$
  (notice that $u$ appears both in $\bar{(G'_1)}$ and $\bar{(G'_2)}$). Thus,  both $\bar{(G'_1)}$ and $\bar{(G'_2)}$ are minors of $\bar{G}$, hence  $\barp{G}$ as well.
  The case where $G'$ is obtained from $G$ by contracting an edge $e$ can be proved using similar but simpler arguments. Indeed, if $e$ is a bridge in $G$, then we have that $\barp{G} = \bar{G}$, and if it is not, it suffices to contract in $\bar{G}$ the edge $\{u,v\}$ with $u,v \in V(\bar{G})$ such that  $e$ is an edge between $T_u$ and $T_v$.

   \medskip
   \textbf{Proof of~\ref{rectreeofbridges}}: Let $v \in V(\bar{G})$ and $T_v$ be its associated tree of bridges in $G$.
   Observe first that we may have $\bar{(G\setminus V(T_v))} \neq \bar{G}\setminus v$. Indeed, if for example we consider
   $G$ composed of two vertex-disjoint triangles $\{a,b,c\}$, $\{a',b',c'\}$ and an edge $e=\{a,a'\}$, and if we consider $T_v = \{e\}$,
   then $\bar{G}\setminus v$ is composed of two disjoint edges, whereas $\bar{(G \setminus V(T_v))}$ is composed of two isolated vertices.
   However, it is easy to verify that $\bar{(G\setminus V(T_v))} = \bar{(\bar{G}\setminus v)}$.
   Let us now prove that $\min_{T} \bd(G \setminus  V(T)) \le \min_{v \in V(\bar{G})} \bd(\bar{G} \setminus  v)$.
   Let $v^*$ be a vertex minimizing $\bd(\bar{G} \setminus  v)$.
  We have $\min_{T} \bd(G \setminus  V(T)) \le \bd(G \setminus  V(T_{v^*})) =  \bd(\bar{(G\setminus V(T_{v^*})})$ using Item~\ref{contractionsstable} in the last equality,
  and  $\bd(\bar{(G\setminus V(T_{v^*})})=\bd(\bar{(\bar{G}\setminus v^*)})=\bd(\bar{G} \setminus  v^*)$ using again Item~\ref{contractionsstable}.

  For the other inequality, let $T^0$ be a tree of bridges that minimizes $\bd(G\setminus V(T))$.
  If $T^0$ is not inclusion-wise maximal, let $T^*$ be any inclusion-wise maximal tree of bridges containing $T^0$.
  Note that as $G\setminus V(T^*)$ is a subgraph of $G\setminus V(T^0)$, by Item~\ref{minorclosed} we get that $\bd(G \setminus  V(T^*)) \le \bd(G \setminus  V(T^0))$,
  implying that $T^*$ also minimizes $\bd(G\setminus V(T))$.
  Let $v^* \in V(\bar{G})$ such that $T_{v^*} = T^*$.
  We have $\min_{v \in V(\bar{G})} \bd(\bar{G} \setminus  v) \le \bd(\bar{G}\setminus v^*) = \bd(\bar{(G\setminus T_{v^*})}) = \bd(G\setminus T_{v^*})$.

  \medskip
  \textbf{Proof of~\ref{removeX}}. We use induction on $|X|$, the base case~$X = \emptyset$ being trivial. For the induction step, pick an arbitrary~$v \in X$, let~$X' := X \setminus \{v\}$, and~$G' := G \setminus X'$. By induction we have~$\bd(G) \leq |X'| + \bd(G')$. Let~$G'_i$ be the connected component of~$G'$ containing~$v$. Using~$v$ as a singleton tree of bridges in~$G'_i$, Item~\ref{rectreeofbridges} shows that~$\bd(G'_i) \leq 1 + \bd(G'_i \setminus v) \leq 1 + \bd(G' \setminus v)$. Since all other components~$G'_j$ of~$G'$ also occur as components of~$G' \setminus v$, it follows that~$\bd(G'_j) \leq \bd(G' \setminus v)$, implying~$\bd(G') \leq 1 + \bd(G' \setminus v) = 1 + \bd(G \setminus X)$ since~$G' \setminus v = G \setminus X$. Hence~$\bd(G) \leq |X'| + 1 + \bd(G \setminus X)$.
	
  \medskip
  \textbf{Proof of~\ref{item:treewidth}}: We use induction on~$|V(G)|$; the base case follows directly from the definitions. It is well-known (cf.~\cite[Lemma 6]{Bodlaender98}) that the tree-width of~$G$ is the maximum tree-width of its biconnected components. Hence it suffices to prove that for an arbitrary biconnected component~$G'$ of~$G$, we have~$\tw(G') \leq \bd(G')$. If~$G'$ consists of a single edge, then~$\tw(G') = \bd(G') = 1$. Otherwise,~$G'$ is a connected bridgeless graph. This implies~$\bar{(G')} = G'$, so by Definition~\ref{def:bd} there is a vertex~$v \in V(G')$ such that~$\bd(G') = 1 + \bd(G' \setminus v)$. Since~$G'$ is a minor of~$G$, we have~$\bd(G') \leq \bd(G)$ by Item~\ref{minorclosed}. By induction, the tree-width of~$G' \setminus v$ is at most~$\bd(G' \setminus v) \leq \bd(G) - 1$. Adding vertex~$v$ to all bags of a tree decomposition of this width, gives a valid tree decomposition of~$G'$ of width at most~$\bd(G' \setminus v) + 1 \leq \bd(G')$. Hence~$\tw(G') \leq \bd(G')$ for all biconnected components of~$G$.
  \end{proof}

A~$(c+1) \times (c+1)$-grid is a planar graph of tree-width exactly~$c+1$~\cite[Cor.~89]{Bodlaender98}, which implies by Item~\ref{item:treewidth} of Proposition~\ref{proposition:bd} that its bridge-depth is larger than~$c$. This gives the following consequence of Proposition~\ref{proposition:bd}, which will be useful when invoking algorithmic meta-theorems.

\begin{observation} \label{obs:boundedbd:family}
For each~$c \in \mathbb{N}$, the graphs of bridge-depth at most~$c$ form a minor-closed family that excludes a planar graph. By the Graph Minor Theorem~\cite{RobertsonS04}, there is a finite set of forbidden minors~$\mathcal{H}_c$ such that~$\bd(G) \leq c$ if and only if~$G$ excludes all graphs of~$\mathcal{H}_c$ as a minor. The set~$\mathcal{H}_c$ contains a planar graph, since some planar graphs have bridge-depth~$> c$.
\end{observation}

Observation~\ref{obs:boundedbd:family}, together with known results on minor testing, imply the following.

\begin{proposition}[{Follows from \cite[Thm.~7.1]{Bodlaender96}}] \label{prop:checkingBridgeDepth}
For each constant~$c \in \mathbb{N}$, there is a linear-time algorithm to test whether the bridge-depth of a given graph~$G$ is at most~$c$.
\end{proposition}

Fomin et al.~\cite[Thm.~1.3]{FominLMPS16} gave a generic approximation algorithm for finding a small vertex set that hits forbidden minors from a finite forbidden set containing a planar graph. By Observation~\ref{obs:boundedbd:family}, deleting vertices to obtain a graph of bounded bridge-depth fits into their framework.

\begin{proposition}[Follows from {\cite[Thm.~1.3]{FominLMPS16}}] \label{prop:approxmod}
For each fixed~$c \in \mathbb{N}$ there is a polynomial-time algorithm that, given a graph~$G$, outputs a set~$X \subseteq V(G)$ such that~$\bd(G \setminus X) \leq c$ and~$|X| \leq \Oh(|X_{\mathrm{opt}}| \log^{2/3} |X_{\mathrm{opt}}|)$, where~$|X_{\mathrm{opt}}|$ is the minimum size of such a set.
\end{proposition}

The following concept will be crucial to facilitate a recursive approach for reducing graphs of bounded bridge-depth.

\begin{definition} \label{def:loweringtree}
  A \emph{lowering tree} $T$ of a graph $G$ is a tree of bridges (possibly consisting of a single vertex and no bridges) such that $\bd(G \setminus V(T)) = \bd(G)-1$.
 \end{definition}

Item~\ref{rectreeofbridges} of Proposition~\ref{proposition:bd} implies that any connected graph $G$ has a lowering tree.

\begin{restatable}{proposition}{propXcomputingTreeDec} \label{prop:computingTreeDec}
For each fixed~$c \in \mathbb{N}$ there is an algorithm that, given a connected graph~$G$ on~$n$ vertices of bridge-depth~$c$, computes a lowering tree in~$\Oh(n^2)$ time.
\end{restatable}
\begin{proof}
Given~$G$, we compute its decomposition into biconnected components, which can be done in linear time taking into account that having bounded bridge-depth implies a linear number of edges~\cite{HopcroftT73a}. From this decomposition, it is straightforward to identify the inclusion-maximal trees of bridges in~$G$. For each tree of bridges~$T$ in~$G$, we can test whether~$\bd(G \setminus V(T)) < c = \bd(G)$ in linear time using Proposition~\ref{prop:checkingBridgeDepth}, and we output~$T$ if this is the case. By Proposition~\ref{proposition:bd}, such a tree~$T$ exists. Since~$G$ is decomposed into at most~$n$ trees of bridges, and we need a linear-time computation for each~$T$, this results in an~$\Oh(n^2)$-time algorithm.
\end{proof}

 \section{Bounded minimal blocking sets imply bounded bridge-depth} \label{sec:triangle-path}
The goal of this section is to prove one direction of Theorem~\ref{thm:bridgedepth:blockingsets}, showing that if~$\mathcal{F}$ has bounded-size minimal blocking sets, then~$\mathcal{F}$ has bounded bridge-depth. As explained in Section~\ref{sec:intro}, we prove this via the intermediate structure of \emph{necklace minors} and show that the bridge-depth of a graph~$G$ can be upper-bounded in terms of the longest necklace contained in it as a minor.

This result can be seen as an analog to the fact that the tree-depth of a graph can be bounded in terms of the length of the longest simple path it contains (as a subgraph or as a minor, which is equivalent for paths). A classical proof of this fact (see~\cite{NeMe12}) is to consider a depth-first search tree of~$G$, bounding the tree-depth of~$G$ by the depth of this tree. However, it does not seem immediate to find a similar bound for bridge-depth.

We therefore follow another approach, inspired by the following alternative proof that the tree-depth is upper-bounded by the length of the longest path (which gives a worse bound). Observe that in a connected graph~$G$, any two longest paths intersect at a vertex. (If they did not, one could combine them to make an even longer path.) Given a connected graph $G$ whose longest path has $t$ vertices, we can bound its tree-depth by~$f(t) := \sum_{i=1}^t i$ as follows. Let~$P$ be a longest path in~$G$. Then the longest path in~$G \setminus V(P)$ has strictly fewer than~$t$ vertices, and by induction the tree-depth of~$G \setminus V(P)$ is at most~$f(t-1)$. From the definition of tree-depth, it follows that the tree-depth of~$G$ is at most~$|V(P)| = t$ larger than that of~$G \setminus V(P)$, so the tree-depth of~$G$ is at most~$f(t)$.

In the case of bridge-depth, where paths are replaced with necklaces contained as minors, we cannot afford to remove the entire set of vertices of the corresponding model of a longest necklace,
as the size of this set  cannot be bounded in terms of the length $t$ of the necklace. To overcome this problem, we will prove in Lemma~\ref{lemma:pack}, similarly to the case of paths, that there cannot be two vertex-disjoint longest necklaces. Then we
resort to the Erd\H{o}s-P{\'o}sa property, which gives us a set of vertices of size $f(t)$ whose removal decreases the maximum length of a longest necklace.
We now formalize these ideas.

\begin{definition}\label{def:necklace}
  For $t \in \mathbb{N}$, the \emph{necklace of length $t$},  denoted by $N_t$, is the multigraph having $t+1$ vertices $\{v_i \mid i \in [t+1]\}$ and two parallel edges
  between $v_i$ and $v_{i+1}$ for  $i \in [t]$.
\end{definition}


\begin{observation}\label{obs:neck}
A simple graph $G$ contains $N_t$ as a minor if and only if $G$ contains $t+1$ vertex-disjoint sets $S_i \subseteq V(G)$ such that
each $S_i$ is connected and, for $i \in [t]$, there are at least two edges between $S_i$ and $S_{i+1}$.
\end{observation}

\begin{definition}
   The \emph{necklace-minor length} of a graph $G$, denoted by $\nm(G)$, is the largest length of a necklace contained in $G$ as a minor, or zero if~$G$ contains no such minor.
\end{definition}

We need to introduce the Erd\H{o}s-P{\'o}sa property for packing and covering minor models. Let $\mc F$ be a finite collection of simple graphs. An \emph{$\mc F$-model}  is an $H$-model for some $H \in \mc F$.
Two $\mc F$-models $M_1$ and $M_2$ are \emph{disjoint} if $V(M_1)\cap V(M_2) = \emptyset$.
Let $\nu_{\mc F}(G)$ be the
 maximum
cardinality of a packing of pairwise disjoint
 $\mc F$-models in $G$, and let $\tau_{\mc F}(G)$ be
the minimum size of a subset $X \subseteq V(G)$ such that $G \setminus X$ has no $\mc F$-model. Clearly,
$\nu_{\mc F}(G) \leq \tau_{\mc F}(G)$. We say that the \emph{Erd\H{o}s-P{\'o}sa property} holds for $\mc F$-models if there exists
a bounding function $f \colon \mathbb{N} \rightarrow \mathbb{N}$ such that, for every graph $G$,  $\tau_{\mc F}(G) \le f(\nu_{\mc F}(G))$.

In the case where $\mc F = \{H\}$ contains a single connected graph $H$, Robertson and Seymour~\cite{RobertsonS86} proved the following result.

\begin{theorem}[Robertson and Seymour~\cite{RobertsonS86}]\label{thm:EP}
Let $H$ be a connected graph. The Erd\H{o}s-P{\'o}sa property holds for $H$-models if and only if $H$ is planar.
\end{theorem}

It is worth mentioning that a tight bounding function when $H$ is planar has been recently obtained by van Batenburg et al.~\cite{tight-EP}.
Theorem~\ref{thm:EP} easily implies the following corollary.

\begin{corollary}\label{cor:ep}
 For every $t \geq 1$, the Erd\H{o}s-P{\'o}sa property holds for $N_t$-models.
\end{corollary}
\begin{proof}
For $t\geq 1$, let ${\mathcal F}_t$ be the set containing all minor-minimal {\sl simple} graphs that contain the necklace $N_t$ as a minor. By definition, a simple graph $G$ contains an $N_t$-model if and only if it contains an ${\mathcal F}_t$-model. Clearly, all the graphs in  ${\mathcal F}_t$ are connected and planar, and it is easy to see that $|{\mathcal F}_t|$ is bounded by a function of $t$. For each~$F \in \mathcal{F}_t$, by Theorem~\ref{thm:EP}  there is a function~$f_F$ such that if~$G$ does not contain~$k$ vertex-disjoint models of~$F$, then all the $F$-models of $G$ can be hit by at most~$f_F(k)$ vertices. This implies that if $G$ does not contain $k$ models of any graph in~$\mathcal{F}_t$, then the union of all hitting sets has size bounded by~$\sum _{F \in \mathcal{F}_t} f_F(k)$, and since~$\mathcal{F}_t$ is finite this is a valid bounding function for $N_t$-models.
\end{proof}


\noindent We denote by $f_{N_t}$ the  bounding function for $N_t$-models given by Corollary~\ref{cor:ep}. In a connected bridgeless graph, each pair of maximum-length necklace models intersect at a vertex:

\begin{lemma}\label{lemma:pack}
  If $G$ is a connected bridgeless simple graph with $\nm(G) = t$, then $\nu_{N_t}(G) = 1$.
\end{lemma}

\begin{proof}
  Suppose for a contradiction that $G$ contains two disjoint models $M^1$ and $M^2$ of $N_t$.
  For $i \in [t+1]$ and $\ell \in [2]$, let $S^{\ell}_i$ be the vertex set of $M^\ell$ given by Observation~\ref{obs:neck}. Note that these $2t+2$ subsets of vertices of $G$ are pairwise disjoint,
  and that for any $i \in [t]$, there are at least two edges between $S^{\ell}_i$ and $S^{\ell}_{i+1}$.
  Since $G$ is bridgeless and connected, it is $2$-edge-connected and by Menger's theorem~\cite[\S~3.3]{Diestel16}~$G$ contains two edge-disjoint paths between any pair of vertices.
	Pick two arbitrary vertices~$x_1 \in M^1, x_2 \in M^2$, and let~$P^1, P^2$ be two edge-disjoint paths between them.
  Consider the subpath $Q^\ell$ of $P^\ell$ between the last vertex of $M^1$ that is visited, until the first vertex of $M^2$.
  Let $Q^\ell=(v^\ell_1,\dots,v^\ell_{q_\ell})$ where $v^\ell_1 \in M^1$ and $v^\ell_{q_\ell} \in M^2$.
  Let $a_\ell$ such that $v^\ell_1 \in S^{1}_{a_\ell}$ and $b_{\ell}$ such that $v^\ell_{q_\ell} \in S^{2}_{b_\ell}$.

  Let us first show that if $t$ is odd, then we can use $Q^1$ to find an $N_{t'}$-model $M'$ for some $t' > t$ by ``gluing'' $M^1$ and $M^2$, leading to a contradiction.
  Let $S = S^{1}_{a_1} \cup V(Q^1) \cup S^{2}_{b_1}$. If $a_1 > \frac{t+1}{2}$ define $A = \{S^{1}_1,\dots,S^{1}_{a_1-1}\}$, and otherwise define $A = \{S^{1}_{a_1+1},\dots,S^{1}_{t+1}\}$.
  Similarly, if $b_1 > \frac{t+1}{2}$ define $B = \{S^{2}_1,\dots,S^{2}_{b_1-1}\}$, and otherwise define $B = \{S^{2}_{b_1+1},\dots,S^{2}_{t+1}\}$. Note that the sets in $A,S,B$ are pairwise disjoint. Since $t$ is odd, it can be easily checked that $M' = A \cup \{S\} \cup B$ is an $N_{t'}$-model in $G$ for some $t' > t$;  see
Figure~\ref{fig:EP}(a) for an illustration.

\begin{figure}[h!tbp]
\begin{center}
\includegraphics[width=0.88\textwidth]{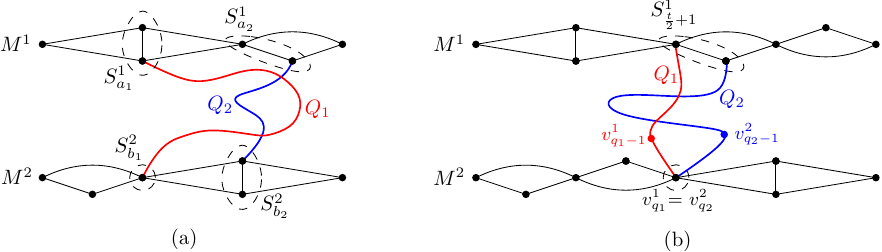}
\end{center}\vspace{-.2cm}
\caption{(a) Example with $t=3$ and $a_1=b_1=2$. (b) Example with $t=4$.}
\label{fig:EP}
\end{figure}

  Let us now consider the case where $t$ is even. Note first that if there exists $\ell \in [2]$ such that $a_\ell \neq \frac{t}{2}+1$ or $b_\ell \neq \frac{t}{2}+1$,
  then we can use $Q^\ell$ to find an $N_{t'}$-model for some $t' > t$ as in the previous case. Hence, it only remains to consider the case where $a_1=b_1=a_2=b_2=\frac{t}{2}+1$, meaning that $Q^1$ and $Q^2$ are two edge-disjoint paths, both between
  $S^{1}_{\frac{t}{2}+1}$ and $S^{2}_{\frac{t}{2}+1}$.
  Let $A = \{S^{1}_1,\dots,S^{1}_{\frac{t}{2}}\}$, $B= \{S^{2}_1,\dots,S^{2}_{\frac{t}{2}}\}$, and $S = S^{1}_{a_1} \cup (V(Q^1) \setminus \{v^1_{q_1}\}) \cup (V(Q^2) \setminus \{v^2_{q_2}\})$.
  We claim that $M' = A \cup \{S,S^{2}_{\frac{t}{2}+1}\} \cup B$ is an $N_{t+1}$-model. Indeed, note in particular there are two edges between
  $S$ and $S^{2}_{\frac{t}{2}+1}$ as we cannot have $v^1_{q_1-1} = v^2_{q_2-1}$ and $v^1_{q_1} = v^2_{q_2}$ because $Q^1$ and $Q^2$ are edge-disjoint and $G$ is a simple graph;  see
Figure~\ref{fig:EP}(b) for an illustration.
\end{proof}

By combining  Corollary~\ref{cor:ep} with Lemma~\ref{lemma:pack} we easily get the following corollary.

\begin{corollary}\label{cor:pack}
Let $G$ be a connected bridgeless graph and  $t = \nm(G)$.
    Then $G$ contains a set of vertices $X$ with $|X| \le f_{N_t}(1)$ such that $\nm(G \setminus X) < \nm(G)$, where $f_{N_t}: \mathbb{N} \to \mathbb{N}$ is the bounding function given by Corollary~\ref{cor:ep}.
\end{corollary}

\begin{proof}
  By Lemma~\ref{lemma:pack}, it follows that $\nu_{N_t}(G) = 1$, and therefore by Corollary~\ref{cor:ep}
  $$
  \tau_{N_t}(G) \ \leq\  f_{N_t}(\nu_{N_t}(G))\ = \ f_{N_t}(1).
  $$
Thus, there exists a set $X \subseteq V(G)$ with $|X| \le f_{N_t}(1)$ such that
 $G \setminus X$ has no $N_t$-model, implying that $\nm(G\setminus X) < t$.
\end{proof}

We are finally in position to prove the following theorem.

\begin{theorem}\label{thm:bd}
There is a function $f \colon \mathbb{N} \to \mathbb{N}$ such that $\bd(G) \le f(\nm(G))$ for all graphs~$G$.
\end{theorem}

\begin{proof}
  We prove the statement by induction on~$\nm(G)$, for the function~$f$ defined by~$f(t) := 1 + \sum_{i=1}^t f_{N_i}(1)$.
  If $\nm(G)=0$, then $G$ is a forest, and by definition of bridge-depth we get $\bd(G) = 1 = f(0)$. Suppose now that $\nm(G)=t$ with $t > 0$.

  Consider the case that~$G$ is connected. Then $\bar{G}$ is also connected and has no bridge, and thus we can apply Corollary~\ref{cor:pack} and get a set $X \subseteq V(\bar{G})$ with $|X| \le f_{N_t}(1)$ such that $\nm(G') < t$, where $G'=\bar{G}\setminus X$.
  By Item~\ref{removeX} of Proposition~\ref{proposition:bd}, we get that $\bd(G) = \bd(\bar{G}) \le |X| + \bd(G')$.
  Let $G'_1, \ldots, G'_{\ell}$ be the connected components of $G'$. As $\nm(G') < t$, we get that $\nm(G'_i) < t$ for every $i \in [\ell]$.
  Then, by induction hypothesis it follows that, for every $i \in [\ell]$ , $\bd(G'_i) \le f(t-1)= 1 + \sum_{i=1}^{t-1}f_{N_i}(1)$.
  Thus, as $\bd(G')=\max_{i \in [\ell]} \bd(G'_i) \leq 1 + \sum_{i=1}^{t-1}f_{N_i}(1)$, we get that
  $$
  \bd(G)\ \le \ |X| + \bd(G') \ \leq \  f_{N_t}(1) + 1 + \sum_{i=1}^{t-1}f_{N_i}(1) \ = \ 1 + \sum_{i=1}^{t}f_{N_i}(1) \ = \ f(\nm(G)).
  $$

  Finally, if $G$ is disconnected, let $G_1, \ldots,G_{\ell}$ be its connected components, and note that $\bd(G)=\max_{i \in [\ell]} \bd(G_i)$. Since for every $i \in [\ell]$ it holds that $\nm(G_i)\leq \nm(G)$, and since the function $f$ is non-decreasing, by applying the above case to each connected component of $G$ we get that
  $$
  \bd(G)\ =\ \max_{i \in [\ell]} \bd(G_i) \ \leq \ \max_{i \in [\ell]}f(\nm(G_i))\ \leq \ \max_{i \in [\ell]}f(\nm(G)) \ = \ f(\nm(G)).\qedhere
  $$
\end{proof}

Now that we established a relation between bridge-depth and necklace minors, our next step is to relate necklace minors to blocking sets.
For this purpose, we use the known triangle-path gadget.

\begin{definition}\label{def:triangle-path}
A \emph{triangle-path of length $t$} is the graph consisting of $t$ vertex-disjoint triangles, with vertex sets $\{\{a_i,b_i,c_i\} \mid i \in [t]\}$, together with the $t-1$ edges $\{ \{b_i,a_{i+1}\} \mid i \in [t-1]\}$.
The \emph{triangle-path-minor length} of a graph $G$, denoted by $\tpm(G)$, is the largest length of a triangle-path contained in $G$ as a minor, or zero if no such minor exists.
\end{definition}

A slight variation of this gadget was used by Fomin and Str\o mme~\cite[Def.~6]{FominS16}. We observe the following (cf.~\cite[Obs.~3--5]{FominS16}).

\begin{observation}\label{obs:mbs-tp}
Let $G$ be a triangle-path of length $t \ge 2$. Then $\mbs(G) \ge t+2$, as $\{a_1,c_1\} \cup \{b_t,c_t\} \cup \{c_i \mid i \in [2,t-1]\}$ is a minimal blocking set.
\end{observation}

\begin{lemma}\label{lemma:tpm-nm}
For any graph $G$, $\tpm(G) \ge \lfloor \frac{\nm(G)+1}{2} \rfloor$.
\end{lemma}
\begin{proof}
  Let $t=\nm(G)$, and let $\{S_i \mid i \in [t+1]\}$ be an $N_t$-model in $G$.
  Let $i \in [\lfloor \frac{t+1}{2} \rfloor]$ and let $e_1=\{u_1,v_1\}$ and $e_2=\{u_2,v_2\}$ be the two edges between $S_{2i-1}$ and $S_{2i}$, with $u_\ell \in S_{2i-1}$ and $v_{\ell} \in S_{2i}$.
  If $u_1 \neq u_2$ then there is a partition $A_1, A_2$ of $S_{2i-1}$ such that $u_i \in A_i$ and $A_i$ is connected for $i \in [2]$,
  and we define $L_i = \{A_1,A_2\}$, $R_i = \{S_{2i}\}$. Otherwise, if  $u_1 = u_2$, then necessarily $v_1 \neq v_2$, and we define symmetrically $L_i = \{S_{2i-1}\}$ and $R_i = \{A_1,A_2\}$.
  In both cases we get that $L_i \cup R_i$ is a model of a triangle, and moreover there is an edge between a vertex in $R_i$ and a vertex in $L_{i+1}$
  for every $i \in [\lfloor \frac{t+1}{2} \rfloor-1]$. This implies that $\bigcup_{i \in [\lfloor \frac{t+1}{2} \rfloor]} (L_i \cup R_i)$ is a model of a triangle-path of length $\lfloor \frac{t+1}{2} \rfloor$ in $G$.
\end{proof}

\begin{corollary}\label{cor:bd-triangle-path}
There is a function $g \colon \mathbb{N} \to \mathbb{N}$ such that~$\bd(G) \le g(\tpm(G))$ for all graphs~$G$.
\end{corollary}
\begin{proof}
   By Lemma~\ref{lemma:tpm-nm}, we have that $\tpm(G) \geq \nm(G)/2$. By letting $g(t):=f(2t)$, where $f$ is the function given by Theorem~\ref{thm:bd}, we get the desired result.
\end{proof}


Note that in the next two results we need the assumption that $\F$ is minor-closed. The following immediate corollary of Corollary~\ref{cor:bd-triangle-path} is critically used in proof of the lower bound given in Theorem~\ref{thm:nokernel}.

\begin{corollary}\label{obs:Funbounded}
Let $\F$ be a minor-closed family of graphs. If $\F$ has unbounded bridge-depth then it contains the family $\F^{{\sf tp}}$ of all triangle-paths.
\end{corollary}

Using this corollary, we can prove one direction of Theorem~\ref{thm:bridgedepth:blockingsets}.

\begin{theorem} \label{thm:largebd:largembs}
Let~$\F$ be a minor-closed family of graphs of unbounded bridge-depth. Then there are graphs in~$\F$ which have arbitrarily large minimal blocking sets.
\end{theorem}
\begin{proof}
By Corollary~\ref{obs:Funbounded}, $\F$ contains all triangle-paths. Since a triangle-path of length~$t$ contains a minimal blocking set of size~$t+2$ by Observation~\ref{obs:mbs-tp}, the theorem follows.
\end{proof}

Theorem~\ref{thm:largebd:largembs} is phrased for graph families, rather than individual graphs. There is no function~$h$ such that~$\bd(G) \leq h(\mbs(G))$ for all~$G$: a bipartite grid graph can have arbitrarily large tree-width and therefore bridge-depth, but its minimal blocking sets have size at most two (cf. Lemma~\ref{lemma:bipartite:blockingsets}).

 \section{Bounded bridge-depth implies bounded-size blocking sets} \label{sec:blockingsets}
 In this section we prove the other direction of Theorem~\ref{thm:bridgedepth:blockingsets}: minimal blocking sets in a graph~$G$ have size at most~$2^{\bd(G)}$. We need the following consequence of K\H{o}nig's theorem.

\begin{restatable}{lemma}{lemmaXmisXwrtXmatching} \label{lemma:mis:wrt:matching}
Let~$G$ be a bipartite graph and let~$M$ be a maximum matching in~$G$. Every maximum independent set of~$G$ contains all vertices that are not saturated by~$M$, and exactly one endpoint of each edge in~$M$.
\end{restatable}
\begin{proof}
Consider a maximum independent set~$S$ in~$G$. Then~$\overline{S} := V(G) \setminus S$ is a minimum vertex cover of~$G$. By K\H{o}nig's theorem (cf.~\cite[Thm. 2.1.1]{Diestel16}) we have~$|\overline{S}| = |M|$. Since~$\overline{S}$ is a vertex cover it contains at least one endpoint of each edge of~$M$; since~$|\overline{S}| = |M|$ it contains exactly one endpoint of each edge of~$M$, and no other vertices of~$G$. So the complement~$S$ contains all vertices that are not saturated by~$M$, and exactly one endpoint of each edge in~$M$.
\end{proof}

The next lemma shows that minimal blocking sets in a bipartite graph have at most two vertices. This was known before, see~\cite[Thm.~14]{HolsK17}. Our self-contained proof highlights an additional property of such minimal blocking sets: the two vertices of minimal blocking sets of size two belong to opposite partite sets. This will be crucial later on.

\begin{restatable}{lemma}{lemmaXbipartiteXblockingsets} \label{lemma:bipartite:blockingsets}
Let~$G$ be a bipartite graph with partite sets~$A$ and~$B$. If~$Y \subseteq V(G)$ is a blocking set in~$G$, then there is a blocking set~$Y' \subseteq Y$ in~$G$ such that one of the following holds:
\begin{itemize}
	\item $|Y'| = 1$, or
	\item $Y' = \{a,b\}$ for some~$a \in A$ and~$b \in B$.
\end{itemize}
\end{restatable}
\begin{proof}
Let~$M$ be a maximum matching in~$G$, let~$V(M)$ be the saturated vertices, and let~$U := V(G) \setminus V(M)$ be the unsaturated vertices. Let~$R_{A \cap U}$ be the vertices that can be reached by a possibly empty $M$-alternating path from~$A \cap U$ (which necessarily starts with a non-matching edge). Let~$R_{B \cap Y}$ be the vertices that can be reached by a possibly empty $M$-alternating path that starts with a {\sl matching edge} from a vertex of~$B \cap Y$. Note that both types of alternating paths move from~$A$ to~$B$ over non-matching edges, and move from~$B$ to~$A$ over matching edges.

We first deal with some cases in which we easily obtain a blocking set~$Y'$ as desired.

(\textbf{Case 1: $A \cap Y \cap R_{A \cap U} \neq \emptyset$}) Let~$a \in A \cap Y \cap R_{A \cap U}$. Then~$a \in A$ can be reached by an $M$-alternating path~$P$ that starts in an unsaturated vertex in the same partite set, implying that~$P$ has even length and ends with a matching edge into~$a$. Hence $M':=M \oplus E(P)$, where $\oplus$ denotes the symmetric difference, is a new maximum matching (which is equal to $M$ if $P$ is empty), and it does not saturate~$a \in A \cap Y$. Lemma~\ref{lemma:mis:wrt:matching} applied to~$M'$ implies that all maximum independent sets of~$G$ contain~$a$, showing that~$Y' := \{a\}$ is a blocking set of size one.

(\textbf{Case 2: $B \cap U \cap R_{B \cap Y} \neq \emptyset$}) By definition, some~$u \in B \cap U$ can be reached by an $M$-alternating path~$P$ that starts in some vertex~$b \in B \cap Y$ that belongs to the same partite set. Similarly as in the previous case, $M':=M \oplus E(P)$ is a new maximum matching (which is equal to $M$ if $P$ is empty) that does not saturate~$b$, so by Lemma~\ref{lemma:mis:wrt:matching} applied to~$M'$ we conclude that~$Y' := \{b\}$ is a blocking set of size one.

(\textbf{Case 3: $A \cap Y \cap R_{B \cap Y} \neq \emptyset$}) By definition, some~$a \in A \cap Y$ is reachable by an $M$-alternating path~$P$ from some~$b \in B \cap Y$, and~$P$ starts with a matching edge. Since it ends in the other partite set, it ends with a matching edge as well; hence both~$a$ and~$b$ are saturated. We claim that~$Y' := \{a,b\}$ is a blocking set in~$G$, as desired. Let~$a = a_1, b_1, \ldots, a_k, b_k = b$ be the vertices on~$P$, so that~$\{a_i, b_i\} \in M$ for all~$i \in [k]$ and~$\{b_i, a_{i+1}\} \in E(G) \setminus M$ for~$i \in [k-1]$. By Lemma~\ref{lemma:mis:wrt:matching}, a maximum independent set in~$G$ contains one endpoint of each of the edges~$\{a_i, b_i\} \in M$. A maximum independent set avoiding~$a_1$ therefore has to contain~$b_1$, preventing it from containing~$a_2$, forcing it to contain~$b_2$, and so on. Hence a maximum independent set avoiding~$a_1$ contains~$b_k$, proving that~$Y' := \{a,b\} = \{a_1, b_k\}$ is a blocking set in~$G$.

(\textbf{Case 4: $B \cap U \cap R_{A \cap U} \neq \emptyset$}) Then some unsaturated vertex of~$A$ can reach an unsaturated vertex of~$B$ by an $M$-alternating path~$P$. But then~$M$ is not a maximum matching since~$M \oplus E(P)$ is larger; a contradiction. Hence this case cannot occur.

\medskip

Assume now that none of the cases above hold. We will conclude the proof of the lemma by deriving a contradiction. Let~$R := R_{A \cap U} \cup R_{B \cap Y}$. The following will be useful.

\begin{claim} \label{claim:reach:x:then:y}
If~$a \in A \cap R$ and~$\{a,b\} \in E(G)$, then~$b \in B \cap R$.
\end{claim}
\begin{claimproof}
By definition,~$a \in A \cap R$ implies~$a$ is reachable by some $M$-alternating path~$P$ that moves to~$A$ over matching edges and moves to~$B$ over non-matching edges, such that~$P$ starts in a vertex~$v \in (A \cup U) \cup (B \cap Y)$. But then~$b$ is also reachable by such an alternating path from~$v$: if~$\{a,b\} \in M$ then, since~$P$ ends at~$a$, edge~$\{a,b\}$ must be the last edge of~$P$, so a prefix of~$P$ is an $M$-alternating path reaching~$b$; if~$\{a,b\} \notin M$ then appending~$\{a,b\}$ to~$P$ yields such an $M$-alternating path. Hence~$b \in R$, and~$b \in B$ follows since~$G$ is bipartite.
\end{claimproof}

Now consider the following set: $S := (A \cap R) \cup (B \setminus R)$.

\smallskip

We will prove that~$S$ is a maximum independent set of~$G$ disjoint from~$Y$, contradicting the assumption that~$Y$ is a blocking set. To see that~$S$ is indeed an independent set, consider any vertex from~$A \cap S$, which belongs to~$A \cap R$. By Claim~\ref{claim:reach:x:then:y} all neighbors of~$a$ belong to~$B \cap R$, and are therefore not contained in~$S$. Hence~$S$ is indeed an independent set. To see that it is maximum, by Lemma~\ref{lemma:mis:wrt:matching} it suffices to argue it contains all of~$U$ and one endpoint of each edge in~$M$.

To see that~$S$ contains all vertices of~$A \cap U$, note that all such vertices are trivially in~$R_{A \cap U}$ and therefore in~$R$, implying their presence in~$A \cap R$ and therefore in~$S$. To see that~$S$ contains all vertices of~$B \cap U$, it suffices to show that~$B \cap U \cap R = \emptyset$, which follows from the fact that neither Case 2 nor Case 4 is applicable. Hence~$S$ contains all vertices of~$U$.

To see that~$S$ contains an endpoint of each edge of~$M$, let~$\{a,b\} \in M$ be arbitrary with~$a \in A$ and~$b \in B$. If~$b \notin R$ then clearly~$b \in S$, as desired. If~$b \in R$, then this is witnessed by an alternating path~$P$ that reaches~$b$ and ends with a non-matching edge. Extending~$P$ with the edge~$\{a,b\} \in M$ then shows that~$a \in R$, so that~$a \in A \cap R$ is an endpoint of the edge contained in~$S$.

Hence~$S$ is a maximum independent set in~$G$. Since Case 1 and Case 3 do not apply, it follows that~$A \cap R \cap Y = \emptyset$, so that~$S \cap A$ contains no vertex from~$Y$. Since all vertices of~$B \cap Y$ are trivially in~$R_{B \cap Y}$ and therefore in~$R$, it follows that~$B \setminus R$ contains no vertex from~$Y$. Hence~$S$ is a maximum independent set in~$G$ disjoint from~$Y$, contradicting the assumption that~$Y$ is a blocking set.
\end{proof}

We will use Lemma~\ref{lemma:bipartite:blockingsets} to power the induction step in the proof of the next theorem, which gives the desired upper-bound on the size of minimal blocking sets in terms of bridge-depth. The main idea in the induction step is as follows. For a connected graph~$G$, we consider a tree of bridges~$T$ for which~$\bd(G \setminus V(T)) < \bd(G)$. We can summarize the relevant ways in which a maximum independent set in~$G$ can be composed out of maximum independent sets for the connected components of~$G \setminus E(T)$, into a weighted tree~$T'$ that is obtained from~$T$ by adding a pendant leaf to each vertex. In turn, maximum-weight independent sets in~$T'$ correspond to maximum independent sets in a bipartite graph obtained from~$T'$ by replacing each vertex by a set of false twins. Applying Lemma~\ref{lemma:bipartite:blockingsets} to this bipartite graph points to two vertices that form a blocking set. We can translate this back into two components of~$G \setminus E(T)$ which are sufficient for constructing a blocking set in~$G$, and apply induction using the fact that~$\bd(G \setminus V(T)) < \bd(G)$.

\begin{restatable}{theorem}{theoremXblockingsetXsize} \label{theorem:blockingset:size}
Let~$G$ be a graph and~$Y_G \subseteq V(G)$ a blocking set in~$G$. There is a blocking set~$Y'_G \subseteq Y_G$ in~$G$ of size at most~$2^{\bd(G)}$.
\end{restatable}
\begin{proof}
We use induction on~$\bd(G) + |V(G)|$. If~$\bd(G) = 1$ then~$G$ is bipartite (it is even a forest, by Proposition~\ref{proposition:bd}) and therefore the claim follows from Lemma~\ref{lemma:bipartite:blockingsets}. For the induction step, assume~$\bd(G) > 1$.

If~$G$ is disconnected, then since a maximum independent set in~$G$ contains a maximum independent set from each connected component, there is a connected component~$C$ such that~$Y_G \cap V(C)$ is a blocking set of~$C$. By induction on~$C$ and~$Y_G \cap V(C)$ we obtain a blocking set~$Y'_G \subseteq Y_G$ for~$C$ of size at most~$2^{\bd(C)} \leq 2^{\bd(G)}$, and~$Y'_G$ is also a blocking set in~$G$.

It remains to deal with the case that~$G$ is connected. If there is a vertex~$v \in V(G)$ such that no maximum independent set in~$G$ contains~$v$, then any blocking set for~$G \setminus v$  is a blocking set in~$G$. Hence we may invoke induction on~$G \setminus v$ and~$Y_G \setminus v$ and output the resulting blocking set~$Y'_G$.

From now on, we assume that each vertex of~$G$ is contained in at least one maximum independent set, and that~$G$ is connected. By Proposition~\ref{proposition:bd}, there is a tree of bridges~$T$ in~$G$ such that~$\bd(G \setminus V(T)) < \bd(G)$. Let~$C_1, \ldots, C_k$ be the connected components of the graph~$G \setminus E(T)$, that is, the graph we obtain by removing all bridges of~$T$ but keeping the vertices incident on them. (We may have~$k=1$ if~$T$ consists of a single vertex.) Since~$T$ is a tree of bridges, every connected component~$C_i$ contains exactly one vertex of~$T$; denote this vertex by~$t_i$.

\medskip

(\textbf{Constructing an auxiliary graph}) To identify a bounded-size blocking set in~$G$, we build an auxiliary bipartite graph~$H$ and vertex subset~$Y_H$, as follows.
\begin{enumerate}
	\item For each~$i \in [k]$, do the following.
	\begin{enumerate}
		\item Add a vertex set~$Z^+_i$ of size~$\alpha(C_i)$ to~$H$. Select~$\alpha(C_i) - \alpha(C_i \setminus Y_G)$ vertices from~$Z^+_i$ arbitrarily, and add them to~$Y_H$.\label{step:zplus}
		\item Add a vertex set~$Z^-_i$ of size~$\alpha(C_i \setminus t_i)$ to~$H$. ($Z^-_i$ may be empty.) Select~$\alpha(C_i \setminus t_i) - \alpha(C_i \setminus t_i \setminus Y_G)$ vertices from~$Z^-_i$ arbitrarily, and add them to~$Y_H$.\label{step:zminus}
		\item Add all possible edges between~$Z^+_i$ and~$Z^-_i$. \label{step:join:plus:minus}
	\end{enumerate}
	\item For each edge~$\{t_i, t_j\}$ of~$T$, add all possible edges between~$Z^+_i$ and~$Z^+_j$.
\end{enumerate}

One can think of~$H$ as being obtained from the tree~$T$ by attaching a degree-1 pendant leaf to each vertex, and then blowing each vertex up into a set of false twins (vertices with the same open neighborhoods), whose size is determined by an independence number. From this interpretation, it is easy to verify that~$H$ is indeed bipartite. Note that for each~$i \in [k]$, all vertices of~$Z^+_i$ belong to the same partite set. Similarly, all vertices of~$Z^-_i$ belong to the same (but opposite) partite set. Intuitively, the bipartite graph~$H$ captures the structure of independent sets in~$G$: for each component~$C_i$ you can choose whether or not to include the attachment point~$t_i$ in your independent set or not. If you do, then you can get the corresponding~$\alpha(C_i)$ vertices from set~$Z^+_i$, but the adjacencies to sets~$Z^+_j$ for~$t_j \in N_T(t_i)$ then prevent you from picking an independent set of size~$\alpha(C_j)$ from~$Z^+_j$ of the neighboring components. If you do not use~$t_i$ in the independent set, you can pick the~$\alpha(C_i \setminus t_i)$ vertices from~$Z^-_i$ instead, which does not impose any restrictions on what you choose for neighboring components. The set~$Y_H$ is chosen in such a way that the loss of having to avoid~$Y_G$ in an independent set in~$G$, corresponds to the loss of having to avoid~$Y_H$ in an independent set for~$H$. We now formalize these ideas.

\begin{claim} \label{claim:bs:gh}
$\alpha(G) = \alpha(H)$ and~$\alpha(G \setminus Y_G) = \alpha(H \setminus Y_H)$.
\end{claim}
\begin{claimproof}
Consider an independent set~$S_G$ in~$G$ (respectively, in~$G \setminus Y_G$). Initialize~$S_H$ as an empty set. For each~$t_i \in V(T)$, if~$t_i \in S_G$ then add~$Z^+_i$ to~$S_H$ (respectively, add~$Z^+_i \setminus Y_H$ to~$S_H$). If~$t_i \notin S_G$, then add~$Z^-_i$ (respectively,~$Z^-_i \setminus Y_H$) to~$S_H$. Observe that by Steps~\ref{step:zplus} and~\ref{step:zminus}, in either case we add at least~$|V(C_i) \cap S_G|$ vertices to~$S_H$, therefore~$|S_H| \geq |S_G|$. To see that~$S_H$ is independent in~$H$, observe that if~$t_i \in S_G$ and we use~$Z^+_i$, then all~$t_j \in N_T(t_i)$ do not belong to~$S_G$ since~$T$ is a subgraph of~$G$, and hence we use the sets~$Z^-_j$ for such~$t_j$. Hence~$\alpha(H) \geq \alpha(G)$. Since the set~$S_H$ avoids~$Y_H$ if~$S_G$ avoids~$Y_G$, we also have~$\alpha(H \setminus Y_H) \geq \alpha(G \setminus Y_G)$.

Consider an independent set~$S_H$ in~$H$ (respectively, in~$H \setminus Y_H)$. For each~$t_i \in V(T)$, by Step~\ref{step:join:plus:minus} of the construction, the set~$S_H$ cannot contain vertices from both~$Z^+_i$ and~$Z^-_i$. If~$S_H \cap Z^+_i \neq \emptyset$, then add a maximum independent set of~$C_i$ (respectively, of~$C_i \setminus Y_G$) to~$S_G$, which has size at least~$|Z^+_i \cap S_H|$. Otherwise, add a maximum independent set of~$C_i \setminus t_i$ (respectively, of~$(C_i \setminus t_i) \setminus Y_G$) to~$S_G$, which has size at least~$|Z^-_i \cap S_H|$. Hence~$|S_G| \geq |S_H|$, and~$S_G$ avoids~$Y_G$ if~$S_H$ avoids~$Y_H$. Using the fact that the components~$C_i$ arose by deleting the edges of the tree~$T$, it is easy to verify that~$S_G$ is independent in~$G$. Hence~$\alpha(G) \geq \alpha(H)$ and~$\alpha(G \setminus Y_G) \geq \alpha(H \setminus Y_H)$.
\end{claimproof}

Since~$Y_G$ is a blocking set in~$G$, Claim~\ref{claim:bs:gh} shows that~$Y_H$ is a blocking set in~$H$. Since~$H$ is bipartite, by Lemma~\ref{lemma:bipartite:blockingsets} there is a blocking set~ $Y'_H \subseteq Y_H$ in~$H$ of size at most two, with the guarantee that if it has size two then its vertices belong to opposite partite sets. Before using~$Y'_H$, we establish a structural claim that will be useful later.

\begin{claim} \label{claim:independence:closedneighbors}
For each~$t_i \in V(T)$ we have~$\alpha(C_i) = \alpha(C_i \setminus N_{C_i}[t_i]) + 1$.
\end{claim}
\begin{claimproof}
Since each vertex of~$G$ is contained in a maximum independent set by the assumption in the beginning of the proof, there is a maximum independent set~$S_G$ of~$G$ containing~$t_i$. Then~$S_G \cap V(C_i)$ is a maximum independent set of~$C_i$ which contains~$t_i$: since~$t_i$ is the only vertex that has neighbors in~$G$ outside of~$C_i$, if there was an independent set in~$C_i$ larger than~$S_G \cap V(C_i)$ then we could substitute it into~$S_G$ to obtain a larger independent set of~$G$, which is impossible. So there is a maximum independent set of~$C_i$ that contains~$t_i$, implying that removing~$t_i$ and its closed neighborhood from~$C_i$ decreases the independence number by exactly one.
\end{claimproof}

(\textbf{Building a blocking set}) Now we build a small blocking set~$Y'_G \subseteq Y_G$ for~$G$, as follows. Initialize~$Y'_G$ as an empty set, and do the following for each vertex~$y \in Y'_H \subseteq Y_H$.

\begin{enumerate}[(I)]
	\item If~$y \in Z^+_i$ for some~$i$ and~$t_i \in Y_G$, then add~$t_i$ to~$Y'_G$. \label{step:ti:blockingset}
	\item If~$y \in Z^+_i$ for some~$i$ but~$t_i \notin Y_G$, then let~$C'_i := C_i \setminus N_{C_i}[t_i] = C_i \setminus N_G[t_i]$. Since~$C'_i$ is a subgraph of~$G \setminus V(T)$, we have~$\bd(C'_i) < \bd(G)$. We claim that~$\alpha(C'_i \setminus Y_G) < \alpha(C'_i)$. To see that, note that~$y \in Z^+_i$ implies that~$Z^+_i \cap Y_H \neq \emptyset$, which implies by Step~\ref{step:zplus} that~$\alpha(C_i) > \alpha(C_i \setminus Y_G)$. By Claim~\ref{claim:independence:closedneighbors}, we know~$\alpha(C'_i) = \alpha(C_i \setminus N_{C_i}[t_i]) = \alpha(C_i) - 1$. If~$\alpha(C'_i \setminus Y_G) = \alpha(C'_i)$, then any independent set of this size in~$C'_i \setminus Y_G$ combines with~$t_i$ to form a maximum independent set in~$C_i$ that is disjoint from~$Y_G$, contradicting the fact that~$\alpha(C_i) > \alpha(C_i \setminus Y_G)$. Hence~$Y_G$ is  indeed a blocking set for~$C'_i$ and we may invoke induction on~$C'_i$ and~$Y_G \cap V(C'_i)$ to obtain a blocking set~$Y^y_G \subseteq Y_G$ of~$C'_i$ of size at most~$2^{\bd(C'_i)} \leq 2^{\bd(G)-1}$. We add~$Y^y_G$ to~$Y'_G$.\label{step:zplus:induction}
	\item If~$y \in Z^-_i$ for some~$i$, then~$Z^-_i \cap Y_H \neq \emptyset$ which implies by Step~\ref{step:zminus} that~$\alpha(C_i \setminus t_i) > \alpha((C_i \setminus t_i) \setminus Y_G)$, hence~$Y_G \cap V(C_i \setminus t_i)$ is a blocking set for~$C_i \setminus t_i$. Since~$C_i \setminus t_i$ is a subgraph of~$G \setminus V(T)$, we have~$\bd(C'_i) < \bd(G)$. Hence we may invoke induction on~$C_i \setminus t_i$ and~$Y_G \cap V(C_i \setminus t_i)$ to obtain a blocking set~$Y^y_G$ for~$C_i \setminus t_i$ of size at most~$2^{\bd(G)-1}$. We add~$Y^y_G$ to~$Y'_G$.\label{step:zminus:induction}
\end{enumerate}

Since~$|Y'_H| \leq 2$, the process above results in a set~$Y'_G \subseteq Y_G$ of size at most~$2^{\bd(G)}$. To complete the proof, it suffices to show that~$Y'_G$ is indeed a blocking set in~$G$.

\begin{claim} \label{claim:independence:number}
$\alpha(G \setminus Y'_G) < \alpha(G)$.
\end{claim}
\begin{claimproof}
Assume for a contradiction that~$\alpha(G \setminus Y'_G) = \alpha(G)$, and let~$S_G$ be a maximum independent set in~$G$ disjoint from~$Y'_G$. We use a similar process as in the proof of Claim~\ref{claim:bs:gh} to build a maximum independent set~$S_H$ in~$H$ disjoint from~$Y'_H$, contradicting the fact that~$Y'_H$ is a blocking set in~$H$.

Initialize~$S_H$ as the empty set. For each~$t_i \in V(T)$, we do the following.
\begin{itemize}
	\item If~$t_i \in S_G$ and~$Y'_H \cap Z^+_i = \emptyset$, then add~$Z^+_i$ to~$S_H$. Note that~$|Z^+_i| \geq |S_G \cap V(C_i)|$ by Step~\ref{step:zplus}.
	\item If~$t_i \notin S_G$ and~$Y'_H \cap Z^-_i = \emptyset$, then add~$Z^-_i$ to~$S_H$. Note that~$|Z^-_i| \geq |S_G \cap V(C_i)|$ by Step~\ref{step:zminus}.
	\item If~$t_i \in S_G$ and~$Y'_H \cap Z^+_i \neq \emptyset$, then we claim that~$|Y'_H \cap Z^+_i| = 1$. This follows from the fact that if~$Y'_H$ has size two, then by Lemma~\ref{lemma:bipartite:blockingsets} its two vertices belong to opposite partite sets, while all of~$Z^+_i$ belongs to the same partite set of~$H$. Since~$t_i \in S_G$ while~$S_G$ avoids~$Y'_G$ and~$Y'_H \cap Z^+_i \neq \emptyset$, Step~\ref{step:ti:blockingset} ensures that~$t_i \notin Y_G$. Hence during the construction of~$Y'_G$ we executed Step~\ref{step:zplus:induction} on account of the unique vertex in~$Y'_H \cap Z^+_i$, which caused a blocking set for~$C'_i := C_i \setminus N_{C_i}[t_i]$ to be added to~$Y'_G$. Since~$S_G$ avoids~$Y'_G$, it follows that~$|S_G \cap V(C'_i)| < \alpha(C'_i)$. As~$|S_G \cap N_{C_i}[t_i]| = 1$ since~$t_i \in S_G$, together with Claim~\ref{claim:independence:closedneighbors} this implies
$$|S_G \cap V(C_i)| = |S_G \cap V(C'_i)| + 1 < \alpha(C'_i) + 1 = \alpha(C_i)  = |Z^+_i|.$$
Now add~$Z^+_i \setminus Y'_H$ to~$S_H$, which has size at least~$|S_G \cap V(C_i)|$ and is disjoint from~$Y'_H$.
	\item If~$t_i \notin S_G$ and~$Y'_H \cap Z^-_i \neq \emptyset$, then similarly as in the previous case we have~$|Y'_H \cap Z^-_i| = 1$, and on account of this vertex we executed Step~\ref{step:zminus:induction} when constructing~$Y'_G$. Hence~$Y'_G$ contains a blocking set for~$C_i \setminus t_i$, and together with the assumption~$t_i \notin S_G$ this implies~$|S_G \cap V(C_i)| \leq \alpha(C_i \setminus t_i) - 1 = |Z^-_i| - 1$. Now add~$Z^-_i \setminus Y'_H$ to~$S_H$, which has size at least~$|S_G \cap V(C_i)|$ and is disjoint from~$Y'_H$.
\end{itemize}

It follows directly from the construction that~$S_H$ is at least as large as~$S_G$ and is disjoint from~$Y'_H$. The fact that~$S_H$ is independent follows for the same reasons as in Claim~\ref{claim:bs:gh}. Since~$S_G$ is a maximum independent set in~$G$, Claim~\ref{claim:bs:gh} then implies  that~$S_H$ is a maximum independent set in~$H$. But this contradicts the fact that~$Y'_H$ is a blocking set in~$H$.
\end{claimproof}

\noindent Claim~\ref{claim:independence:number} shows that~$Y'_G$ is a blocking set in~$G$, which concludes the proof of Theorem~\ref{theorem:blockingset:size}.
\end{proof}

Note that Theorem~\ref{theorem:blockingset:size} and Theorem~\ref{thm:largebd:largembs} together prove Theorem~\ref{thm:bridgedepth:blockingsets}. We finish the section by showing that the upper-bound of~$2^{\bd(G)}$ on the size of minimal blocking sets is tight.

\begin{theorem} \label{thm:tightness}
For every~$c \in \mathbb{N}$, there is a graph~$G$ with~$\bd(G) \leq c$ that contains a minimal blocking set of size~$2^c$.
\end{theorem}

\begin{proof}
%
Recall the notion of triangle-path from Definition~\ref{def:triangle-path}. For~$t \geq 2$, let a \emph{truncated triangle-path of length~$t$} be the graph~$U_t$ obtained from a triangle-path of length~$t$ by removing vertices~$a_1$ and~$b_t$; see Figure~\ref{fig:truncated}. Analogously to Observation~\ref{obs:mbs-tp}, we show that~$Y_t := \{c_i \mid i \in [t]\}$ is a minimal blocking set in~$U_t$. Since~$Y_t$ is an independent set of size~$t$, while (the remainders of) the triangles in~$U_t$ partition the vertices of~$U_t$ into~$t$ cliques, it follows that~$\alpha(U_t) = t$. The set~$Y_t$ is a blocking set, since~$U_t \setminus Y_t$ is a path on~$2(t-1)$ vertices, whose independence number is only~$t-1$. Finally, it is easy to see that for any~$y \in Y_t$, there is a size-$t$ independent set in~$U_t \setminus (Y_t \setminus y)$ that consists of the vertex~$y$ and, for every (remainder of a) triangle in~$U_t$, the vertex closest to~$y$.

\begin{figure}[htb]
\begin{center}
\includegraphics{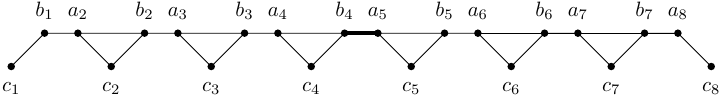}
\end{center}\vspace{-.2cm}
\caption{Truncated triangle path~$U_8$ of length~$8$, illustrating Theorem~\ref{thm:tightness}. Removing the fat middle bridge and its incident vertices, leaves two connected components isomorphic to~$U_4$.}
\label{fig:truncated}
\end{figure}

Hence~$U_t$ has a minimal blocking set of size~$t$, for all~$t \geq 2$. To prove the theorem, it therefore suffices to show that~$\bd(U_{2^c}) \leq c$ for all~$c \in \mathbb{N}$. We prove this by induction on~$c$. For~$c=1$, note that the graph~$U_2$ is just the four-vertex path. Hence it is a forest, implying~$\bd(U_2) = 1$ by Proposition~\ref{proposition:bd}. For~$c > 1$, consider the graph~$U_{2^c}$. By construction, the middle edge~$e = \{b_{2^{c-1}}, a_{2^{c-1}+1}\}$ is a bridge in~$U_{2^c}$. Let~$T$ be the tree in~$U_{2^c}$ consisting of the single bridge~$e$. Note that removing~$V(T)$ splits~$U_{2^c}$ evenly, into two connected components that are both isomorphic to~$U_{2^{c-1}}$. By induction,~$\bd(U_{2^{c-1}}) \leq c-1$. Then Proposition~\ref{proposition:bd} shows that~$\bd(U_{2^c}) \leq 1 + \bd(U_t \setminus V(T)) = 1 + (c-1) = c$.
\end{proof}

 \section{Kernelization for modulators to bounded bridge-depth} \label{sec:kernel:summary}
 To establish the positive direction of Theorem~\ref{thm:characterization}, we develop a polynomial kernel for \textsc{Vertex Cover} parameterized by the size of a modulator~$X$ whose removal leaves a graph of constant bridge-depth; an approximately optimal such set~$X$ can be computed using Proposition~\ref{prop:approxmod}. As the kernelization is technical and consists of many different reduction rules, with a nontrivial size analysis,  we first present below the high-level idea behind the kernelization and the role of bridge-depth.


\subparagraph{High-level ideas of the kernelization algorithm.} Consider an instance~$(G,k)$ of \textsc{Vertex Cover} with a modulator~$X$ such that~$\bd(G \setminus X)$ is bounded by the constant given by the graph class $\F$. As explained in the introduction, using the fact that minimal blocking sets for the components~$C$ of~$G \setminus X$ have bounded size, the number of such components can easily be bounded by~$|X|^{\Oh(1)}$. To bound the size of individual components, the definition of bridge-depth ensures that in each connected component~$C$ of~$G \setminus X$ there is a tree of bridges~$T \subseteq E(C)$ (called a \emph{lowering tree}) such that removing the vertex set~$V(T)$ from~$C$ decreases the bridge-depth of~$C$. By designing new problem-specific reduction rules, we shrink the tree of bridges to size polynomial in the parameter. This is where the main technical work of the kernelization step lies. It properly subsumes the earlier kernelization for the parameterization by distance to a forest, which is imported as a black box in several previous works~\cite{BougeretS18,FominS16,HolsK17,HolsKP19,MajumdarRR18}. Having bounded the number of components of~$G \setminus X$, together with the size of a lowering tree of bridges in each component, we now proceed as follows: in each component~$C$ of~$G \setminus X$ we move the vertices from a lowering tree of bridges into the set~$X$. This blows up~$|X|$ by a polynomial factor, but strictly decreases the bridge-depth of the graph~$G \setminus X$. We then recursively kernelize the resulting instance. When the bridge-depth of~$G \setminus X$ reaches zero, the graph~$G \setminus X$ is empty and the kernelization is completed. Full details are given below. We start with a formal description of the algorithm in Section~\ref{sec:formalAlgo} (cf. Algorithm~\ref{algo:kernel}), where we also  explain how the remainder of this section is organized.

\medskip

The negative direction of Theorem~\ref{thm:characterization}, presented in Section~\ref{sec:negative-kernel},  is much easier to establish. Using the fact that a minor-closed family~$\mathcal{F}$ of unbounded bridge-depth contains all triangle paths, a kernelization lower bound for modulators to such~$\mathcal{F}$ follows easily using known gadgets.

\subsection{Formal description of the kernelization algorithm}
\label{sec:formalAlgo}

For each integer~$c \in \mathbb{N}$, we will obtain a kernel for the following parameterized problem.

\defparproblem{\textsc{Independent Set} with $c$-bridge-depth modulator (\pb)}
{Undirected graph~$G$, integer~$k$, and a set~$X \subseteq V(G)$ such that~$\bd(G \setminus X) \leq c$.}
{$|X|$.}
{Does~$G$ have an independent set of size~$k$?}

Note that, for structural graph parameterizations that do not refer to the solution size, the \IS and \VC problems are equivalent: an instance~$(G,X,k)$ of \IS with a structural parameter~$|X|$ is equivalent to the similarly-parameterized instance~$(G,X,|V(G)| - k)$ of \VC, and this reduction preserves the parameter. Similarly as in previous work~\cite{JansenB13,BougeretS18}, we therefore present the kernelization algorithm for \IS, since it makes some arguments more intuitive.

We may assume that a modulator~$X$ is given in the input, since one may use the polynomial-time approximation algorithm given by Proposition~\ref{prop:approxmod} to obtain a modulator that can be used to compute the kernel.  This is a standard assumption also used in related work; see~\cite[\S 2.2]{FellowsJR13} for a detailed discussion.

Given an input $(G,X,k)$ of \pb, we henceforth denote by $R: = V(G) \setminus X$ the \emph{remaining} bounded-bridge-depth graph that results from removing the modulator. The kernelization algorithm is defined in Algorithm~\ref{algo:kernel}; throughout this section we present a series of definitions and results that will eventually lead to a proof that this algorithm indeed computes a polynomial kernel for \pb.

Let us first explain the roadmap of the proof, keeping in mind the outline given in the beginning of this section.
Section~\ref{subsec:nbcc} corresponds to the easy part of the kernel, as we invoke the classical machinery to bound the number of connected components of $R$ by a polynomial in the parameter~$|X|$. We employ the notions of \emph{conflict} and \emph{chunk} that were introduced in~\cite{JansenB13} for analyzing how a choice of an independent set from the modulator~$X$ affects the number of additional vertices from~$R$ that can be added to the independent set. The fact that having bounded minimal blocking set size (in $R$) allows for an efficient reduction of the number of connected components of~$R$ (cf. Lemma~\ref{lemma:nbcc}) is implicit in previous work~\cite[\S 3.1]{JansenB13} \cite[Rule 3]{BougeretS18}; see~\cite[Thm.~3]{HolsKP19} for an explicit argument. We use the same kind of arguments in our Rules~\ref{rule:free} and~\ref{rule:chunkdegree}. While some earlier work~\cite{BougeretS18} uses an \emph{annotated} version of the problem, in this work we introduce the notion of \emph{almost-free} set (cf. Definition~\ref{def:almost-free}) to avoid having annotations. This allows us to work in a conceptually cleaner setting.

Let us now turn to Section~\ref{subsec:size}. When working with a modulator to bounded \emph{tree-depth}~\cite{BougeretS18}, once the number of connected components of $R$ is bounded by a polynomial in $|X|$, we can move a tree-depth decreasing vertex from each connected component into the modulator, and get a $(c-1)$-tree-depth-modulator $X'$, with $|X'|$ being still polynomial in $|X|$. This facilitates a recursion on $c$ presented in Section~\ref{sec:recursion}, immediately leading to a polynomial kernel. In our case, we cannot move an entire lowering tree (cf. Definition~\ref{def:loweringtree}) of a connected component of~$R$ into the modulator, as its size may be unbounded.
Thus, the main challenge to obtain the polynomial kernel is addressed in Section~\ref{subsec:size}, where we define new rules to shrink the size of lowering trees. Notice that shrinking these lowering trees generalizes the reduction rules for \textsc{Independent Set} parameterized by distance to a forest~\cite{JansenB13}, since for any component in~$R$ that is a tree, its unique lowering tree consists of the entire component. This explains why we cannot simply use the kernelization for the parameterization by feedback vertex set~\cite{JansenB13} as a black box, as it was the case in previous work~\cite{BougeretS18,FominS16,HolsK17,HolsKP19,MajumdarRR18}.

Let us now formally define the above notions.

\begin{definition}\label{def:conflict}
	For a graph~$G$ and disjoint vertex sets~$X', R' \subseteq V(G)$, we define the number of \emph{conflicts induced on~$R'$ by~$X'$} as
 $$
 \conf^G_{R'}(X') = \alpha(G[R'])-\alpha(G[R'\setminus N_G(X')]).
 $$
When $G$ is clear from the context, $\conf^G_{R'}(X')$ will be simply denoted by $\conf_{R'}(X')$.
\end{definition}
Intuitively,~$\conf_{R'}(X')$ measures how much smaller the independence number of~$G[R']$ becomes when one is forbidden from picking vertices that are adjacent in~$G$ to~$X'$. This allows us to reason about which subsets~$X'$ are viable candidates for occurring in a maximum independent set of~$G$.
\begin{definition}\label{def:almost-free}
Let $(G,X,k)$ be an input of \pb and~$R := G \setminus X$.
\begin{itemize}
\item  A \emph{chunk} is an independent set $X' \subseteq X$ of size at most $2^c$ in~$G$.
We denote by $\X$ the set of all chunks.
\item The \emph{degree} of a chunk $X'$, denoted~$d(X')$, is defined as the number of connected components~$R'$ of~$R$ for which~$\conf_{R'}(X') \neq 0$.
\item A set $Z \subseteq V(R)$ is \emph{free} if~$\conf_{Z}(X')=0$ for all~$X' \in \X$.
\item A set $Z \subseteq V(R)$ is \emph{$x$-almost-free} for~$x \in \mathbb{N}$, if for any $X' \in \X$ such that $\conf_{Z}(X') \neq 0$, it holds that~$\conf_R(X') \ge x$.
\end{itemize}
\end{definition}

Slight abusing notation, for a subgraph~$R'$ of~$R$ we will sometimes write~$\conf_{R'}(X')$ as a shorthand for~$\conf_{V(R')}(X')$.

Note the subtle difference in the subscripts of~$\conf$ for the definition of almost-free. Intuitively, if~$Z$ is $x$-almost-free, then any chunk that makes a conflict on the subgraph induced by~$Z$, will make at least~$x$ conflicts on the entire bounded bridge-depth graph~$R$. For~$x \geq |X|$, this will allow us to infer that there is a maximum independent set of~$G$ that does not contain any chunk that makes a conflict on~$Z$.

\begin{lemma}\label{lemma:conflictandcopoly}
Let $(G,X,k)$ be an instance of \pb. There is a polynomial-time algorithm to compute the function $\conf$ for any~$R' \subseteq V(R)$, to compute the degree~$d$ of a chunk, and to decide if a subset $Z \subseteq V(R)$ is free or $x$-almost free.
\end{lemma}
\begin{proof}
For any~$R' \subseteq V(R)$ we have~$\tw(G[R']) \leq \bd(G[R']) \leq c$, by Item~\ref{item:treewidth} of Proposition~\ref{proposition:bd}. Since \IS is fixed-parameter tractable parameterized by treewidth (cf.~\cite[\S 7.3.1]{CyganFKLMPPS15}), and a tree decomposition of constant width~$c$ can be computed in linear time~\cite{Bodlaender96}, this means independence numbers of subgraphs of~$G \setminus X$ can be computed in polynomial time. Since the number of potential chunks is polynomial in~$|X|$ since~$c$ is a constant, this allows all mentioned quantities to be efficiently computed.
\end{proof}

The following lemma shows why having small blocking sets is useful to characterize the interaction between maximum independent sets in $R$ and chunks.

\begin{lemma}\label{lemma:smallkiller}
Let $(G,X,k)$ be an instance of \pb and let $R' \subseteq V(R)$.
For every independent set $S_X \subseteq X$ such that $\conf_{R'}(S_X) \neq 0$ there exists a chunk $X' \in \X$, with $X' \subseteq S_X$, such that $\conf_{R'}(X') \neq 0$.
\end{lemma}
\begin{proof}
Let $Z = N_G(S_X) \cap R'$. The fact that $\conf_{R'}(S_X) \neq 0$ implies that $Z$ is a blocking set of $G[R']$.
As $\bd(R') \le c$, which implies that $\mbs(R') \le 2^c$ by Theorem~\ref{theorem:blockingset:size}, there exists $Z' \subseteq Z$ with $|Z'| \le 2^c$ such that $Z'$ is a blocking set of $R'$.
Thus, with any $z \in Z'$ we associate a vertex $v_z \in S_X$ such that $v_z$ is adjacent to $z$, and we get that $X' := \{v_z \mid z \in Z'\}$ is the desired chunk.
\end{proof}

\begin{corollary}\label{corollary:free}
Let $(G,X,k)$ be an instance of \pb. A set $Z \subseteq V(R)$ is free if and only if for any independent set $S \subseteq X$, $\conf_Z(S)=0$ (or equivalently $\alpha(G[Z \setminus N_G(S)])=\alpha(G[Z])$).
\end{corollary}

In Algorithm~\ref{algo:kernel} we present the pseudo-code of the kernelization algorithm. The reduction rules employed by the algorithm will be presented in the subsequent sections.

\SetKwInput{KwData}{Input}
\SetKwInput{KwResult}{Output}


\medskip
\begin{algorithm}[!ht]
\SetEndCharOfAlgoLine{}
 \KwData{an input $(G,X,k,c)$ of \pb}
 \KwResult{an equivalent instance of size polynomial (for fixed $c$) in $|X|$ (see Theorem~\ref{thm:kernel:is})}

	\vspace{.2cm}

 \If{$c=0$}{
return $G$
}

 If Rule 1 can be applied to $(G,X,k)$, apply it and restart from Line 1

 If Rule 2 can be applied to $(G,X,k)$, apply it and restart from Line 1

 \tcc*{$\cc(G[R])$ is now bounded according to Lemma~\ref{lemma:nbcc}}

 Define $X_1 = \emptyset$

\For{each connected component $R'$ of $R$}{
    Compute a lowering tree $T$ \tcc*{done in polynomial time by Proposition~\ref{prop:computingTreeDec}} \label{algo:T}

    Compute a longest path $T'$ in $T$ \label{algo:T'}

    If Meta-Rule~\ref{rule:*A} can be applied to $((G,X,k), T')$, apply it and restart from Line 1

    If Meta-Rule~\ref{rule:B2B2} can be applied to $((G,X,k), T')$, apply it and restart from Line 1 \label{algo:diambounded}

    \tcc*{$T$ now has bounded diameter $\diam(T)$ according to Lemma~\ref{lemma:diameter}}

    If Meta-Rule~\ref{rule:deg} can be applied to $((G,X,k), T)$, apply it and restart from Line 1 \label{algo:degreebounded}

   \tcc*{$T$ now has bounded max.\ degree $\Delta(T)$ according to Lemma~\ref{obs:degree}}

   \tcc*{It now remains to bound the number of leaves of $T$}

   Define $T_1$ by removing from $T$ any $A$-leaf that has a $B$-parent

   \tcc*{$T_1$ has no $A$-leaf with $B$-parent, $\Delta(T_1) \le \Delta(T)$ and $\diam(T_1) \le \diam(T)$ according to Lemma~\ref{lemma:absorb}}




    If Meta-Rule~\ref{rule:*A} can be applied to $((G,X,k), T_1)$, apply it and restart from Line 1 \label{algo:n*abounded}

   \tcc*{$n_{{\sf *A}}(T_1)$ is now bounded according to Lemma~\ref{lemma:n*A}}

    If Meta-Rule~\ref{rule:Bleaf} can be applied to $((G,X,k), T_1)$, apply it and restart from Line 1 \label{algo:nBBBbounded}

   \tcc*{$n_{{\sf BBB}}(T_1)$ is now bounded according to Lemma~\ref{lemma:BBBleaf}}

    If Meta-Rule~\ref{rule:BB2leaf} can be applied to $((G,X,k), T_1)$, apply it and restart from Line 1 \label{algo:nBB2bounded}

    \tcc*{$n_{{\sf BB2}}(T_1)$ is now bounded according to Lemma~\ref{lemma:BB2leaf}}

    \tcc*{$|V(T)|$  is now bounded according to Lemma~\ref{lemma:Tbounded}}

   $X_1 : = X_1 \cup V(T)$
 }

 Define $X_2 = X \cup X_1$ \label{algo:end}

 \tcc*{$|X_2|$ is bounded and is a $(c-1)$-\bd-modulator}

 Return Algorithm~\ref{algo:kernel}$(G,X_2,k,c-1)$

 \caption{A polynomial kernel for \pb}
\label{algo:kernel}
\end{algorithm}

\subsection{Bounding the number of connected components}\label{subsec:nbcc}


Let us first define some rules that will be used in Algorithm~\ref{algo:kernel}. In these rules, we let $k'$ denote the desired independent set size of the resulting instance.

\begin{ruleN}\label{rule:free}
Let $(G,X,k)$ be an input of \pb.
If there exists a connected component $R'$ of $R$ such that $V(R')$ is free, then delete $V(R')$ from~$G$ and define $k'=k- \alpha(R')$.
\end{ruleN}
\begin{ruleproof}
Let $G'$ be the graph obtained after applying the rule.
We just  prove  that $\alpha(G') \ge k'$ implies $\alpha(G) \ge k$, as the other implication is straightforward.
Let $S'$ be an independent set of $G'$ with $|S'| \ge k'$.
Recall that for any subset $V'$ we use $S'_{V'}$ to denote $S' \cap V'$.
As $V(R')$ is free, according to Corollary~\ref{corollary:free} we get that $\conf_{R'}(S'_X)=0$, implying that there exists a maximum independent set
$Z$ of $R'$ such that $S' \cup Z$ is independent in $G$, implying the desired inequality.
\end{ruleproof}

\begin{ruleN}\label{rule:chunkdegree}
Let $(G,X,k)$ be an input of \pb.
If there exists a connected component $R'$ of $R$ such that for any chunk $X'$ with~$\conf^G_{R'}(X') \neq 0$ we have that $d(X') \geq |X|+1$, then delete all edges from $X$ to $R'$ and define $k'=k$.
\end{ruleN}
\begin{ruleproof}
Let~$G'$ be the graph obtained by applying the rule; then~$\alpha(G) \geq k$ trivially implies~$\alpha(G') \geq k$. Let us prove that $\alpha(G') \ge k$ implies that  $\alpha(G) \ge k$.
Let $S'$ be an independent set of $G'$ of size at least $k$.
If $\conf^G_{R'}(S'_X) = 0$ then there exists a maximum independent set $S^*_{R'}$ of $G[R']$ such that $S'_X \cup S^*_{R'}$ is still independent in $G$.
Thus, in this case we can replace $S'_{R'}$ by $S^*_{R'}$ in $S'$ and get an independent set of $G$ of size at least $|S'|$.
Otherwise, $\conf^G_{R'}(S'_X) \neq 0$ implies, by Lemma~\ref{lemma:smallkiller}, that there exists a chunk $X' \subseteq S'_X$ such that
$\conf^G_{R'}(X') \neq 0$. Since the rule applied, we know that $d(X') \geq |X|+1$, implying that $\conf_{R \setminus R'}(S'_X) \ge \conf_{R \setminus R'}(X') \ge d(X')-1 \ge |X|$.
Thus, consider any maximum independent set $S^*_{R \setminus R'}$ of $R \setminus R'$ and any maximum independent set $S^*_{R'}$ of $R'$.
We have $|S^*_{R \setminus R'}| \ge |S'_X \cup S'_{R\setminus R'}|$, implying that $S^*_{R \setminus R'} \cup S^*_{R'}$ is an independent set of $G$ of size at least $|S'|$.
\end{ruleproof}

Notice that applying Rule~\ref{rule:chunkdegree} may trigger the application of Rule~\ref{rule:free}.
Observe also that applying Rule~\ref{rule:chunkdegree} does not guarantee that $d(X') \le |X|$ for any chunk $X'$.
However, as proved in the next lemma, applying Rules~\ref{rule:free} and~\ref{rule:chunkdegree} exhaustively is sufficient to bound the number of connected components.



\begin{lemma}\label{lemma:nbcc}
Let $(G,X,k)$ be an input of \pb.
Let $(G',X',k')$ be the input of \pb obtained after applying exhaustively Rule~\ref{rule:free} and Rule~\ref{rule:chunkdegree}.
Recall that $R = V(G)\setminus X$, and let $\tilde{R}=V(G')\setminus X$.
Then,
\begin{itemize}
  \item $X'=X$ and $k' \le k$,
  \item $G'$ is a subgraph of $G$ and $\tilde{R} \subseteq R$, and
  \item $\cc(G'[\tilde{R}]) \le |\X|\cdot |X|$.
\end{itemize}
\end{lemma}
\begin{proof}
  The claims in the first two items hold as, given an input $(G,X,k)$ of \pb, Rule~\ref{rule:free} and Rule~\ref{rule:chunkdegree} either delete
  connected components of $R$, or delete edges between $X$ and $R$.
  As $X'=X$, the set of chunks in $(G',X,k)$ is also $\X$.
  Let us now turn to the last claim.
Let $R'$ be a connected component of $G'$.
As Rule~\ref{rule:free} cannot be applied, by Corollary~\ref{corollary:free} 
we have that $\{X' \in \X \mid \conf^{G'}_{R'}(X') \neq 0 \} \neq \emptyset$.
As Rule~\ref{rule:chunkdegree} cannot be applied, we can associate with $R'$ a chunk $X^{R'} \in \X$ such that $\conf^{G'}_{R'}(X^{R'}) \neq 0$ and $d(X^{R'}) \le |X|$. Thus, if $\cc(G'[\tilde{R}]) > |\X| \cdot |X|$, by the pigeonhole principle more than $|X|$ different connected components will be associated with the same chunk $X'$, contradicting the fact that $d(X') \le |X|$.
\end{proof}

We remark that by using a more sophisticated marking scheme for connected components based on maximum matchings in an auxiliary bipartite graph~\cite[Thm 3]{HolsKP19}, it is possible to reduce the number of connected components to~$|\mathcal{X}|$. We have opted for a simpler approach since it suffices to obtain polynomial kernels that complete the dichotomy.

\subsection{Bounding the size of a lowering tree of each connected component}\label{subsec:size}
Before formally defining the tools we need, let us explain the common ideas behind all the Meta-Rules (except Meta-Rule~\ref{rule:deg}) of Sections~\ref{subsubsec:diam} and~\ref{subsubsec:degleaves}, which
are designed to shrink a lowering tree $T$ of a connected component $R'$ of $R$.
Roughly speaking, a $T$-conflict structure $C$ is formed by a subset of vertices of~$T$ together with the connected components of~$R - E(T)$ that they are contained in (see Definition~\ref{def:conflictstrucutre}), which has a simple structure that allows us to shrink~$T$ locally provided that $C$ does not interact too much with the modulator~$X$. The condition of not interacting too much is captured by the notion of $C$ being almost free in our formalism (more precisely $(|X|+\Delta)$-almost-free where $\Delta$ is a constant).
This explains why all Meta-Rules below (except Meta-Rule~\ref{rule:deg}) take the following form:
\begin{quote}
 ``If there is a $T$-conflict structure $C$ that is $(|X|+\Delta)$-almost-free, then reduce $T$ (by performing local modification around $C$).''
\end{quote}
Suppose now that such a rule does not apply anywhere on a lowering tree $T$ of a connected component $R'$ of $R$.
The idea to bound the size of $T$ is as follows. If by contradiction $T$ is very large, then we can find a partition $\P = \{V^T_i \mid i \in [|\P|]\}$ of $T$
(and thus a partition of $R'$ by defining $V^{R'}_i = \bigcup_{v \in V^T_i}H^T_v$ and $\P^{R'}=\{V^{R'}_i \mid i \in [|\P|]\}$, where $H^T_v$ is defined in Definition~\ref{def:pending}) such that many of the parts $V' \in V^{R'}_i$ are $T$-conflict structures.
As the Meta-Rule cannot be applied on any of these parts, it implies that none of them is $(|X|+\Delta)$-almost-free. Now we conclude by using the definition of almost-free as follows.
If a part $V'$ is not $(|X|+\Delta)$-almost-free, it means that $V'$ has a ``private'' chunk $X'_{V'} \in \X$ such that $\conf_{V'}(X'_{V'}) \neq 0$ and $\conf_{R'}(X'_{V'}) < |X|+\Delta$.
By a pigeonhole argument, if $T$ is too large, then there will exist a chunk $X'$ that is the private chunk of many (say $x \ge |X|+\Delta$) different $V'$.
However, and this is where the notion of \emph{$\alpha$-additive partition} (see Definition~\ref{def:additive-partition}) comes into play, as we will define $\P^{R'}$ such that it is an $\alpha$-additive partition,
such a chunk will have $\conf_{R'}(X') \ge x$, contradicting the fact that  $\conf_{R'}(X') < |X|+\Delta$.
The notion of \emph{types} (see Definition~\ref{def:pending}) also plays a crucial role in our approach.
Indeed, given a lowering tree $T$ (or any tree of bridges) of a connected component $R'$ of $R$, types allow us to capture in a simple way the condition we need in $T$ and in the relations
between $T$ and $R' \setminus V(T)$, while hiding the complex structure of $G[R' \setminus V(T)]$. For example, a conflict structure of type~$1$ only needs two adjacent (in $T$) vertices $u$ and $v$ where at least one has type $A$, regardless of the exact structure of $H^T_u$ and $H^T_v$.

\begin{definition}\label{def:pending}
Let $R$ be a graph, let $T$ be a tree of bridges of $R$, and let $R^* = R \setminus E(T)$.
For a vertex $v \in V(T)$, the \emph{pending component of $v$ in $T$}, denoted by $H^T_v$, or simply $H_v$ when clear from context, is the connected component of $R^*$ containing $v$. We call $v$ the \emph{root} of $H_v$.
We say that $v$ has
\begin{itemize}
\item \emph{type A in $T$} if $\alpha(R[H_v])=\alpha(R[H_v \setminus \{v\}])$, and
\item \emph{type B in $T$} if $\alpha(R[H_v])=\alpha(R[H_v \setminus \{v\}])+1$.
\end{itemize}
\end{definition}

Notice that for any distinct vertices $u,v \in V(T)$ we have $H_u \cap H_v = \emptyset$. Moreover, since the pending components arise by removing a tree of bridges, it follows that the only vertex of~$H_u$ that can have neighbors outside~$H_u$ is the root~$u$ itself.

When working with pending components, like in the following definition, we will sometimes refer to the vertex set of a component~$H_v$ simply as~$H_v$, if there is no risk of confusion.


\begin{definition}\label{def:conflictstrucutre}
Let $R$ be a graph and let $T$ be a tree of bridges of $R$.
We say that a subset of vertices $C \subseteq V(R)$ is a \emph{$T$-conflict structure}
\begin{itemize}
\item \emph{of type 1} if $C=H_{v_1} \cup H_{v_2}$ where $\{v_1,v_2\} \in E(T)$, and $v_1$ or $v_2$ has type $A$ in $T$,
\item \emph{of type 2} if $C=H_{v_1} \cup H_{v_2}$ where
  \begin{itemize}
   \item there exists $u_1, u_2 \in V(T)$ such that $(u_2,v_1,v_2,u_1)$ is a path in $T$,
   \item $d_T(v_1)=d_T(v_2)=2$, and
   \item $v_1$ and $v_2$ have type $B$ in $T$,
   \end{itemize}
\item \emph{of type 3} if $C=H_u$ where $u$ is a leaf of $T$ and has type $B$,
\item \emph{of type 4} if $C=H_{v_1} \cup H_{v_2}$ where
  \begin{itemize}
   \item there exists $u \in V(T)$ such that $(v_1,v_2,u)$ is a path in $T$,
   \item $v_1$ is a leaf in $T$ and $d_T(v_2)=2$, and
   \item $v_1$ and $v_2$ have type $B$ in $T$.
   \end{itemize}
\end{itemize}
\end{definition}

The following lemma holds by the same arguments used in the proof of Lemma~\ref{lemma:conflictandcopoly}.

\begin{lemma}\label{lemma:typesandcopoly}
Let $(G,X,k)$ be an instance of \pb and $T$ be a tree of bridges of a connected component $R'$ of $R$.
Computing the type of a vertex $v \in V(T)$, as well as deciding if there exists a $T$-conflict structure $C$ of any fixed type, can done in polynomial time.
\end{lemma}

The following definition will be useful to argue that a chunk makes a large number of conflicts on~$R$, which will in turn allow us to argue that a reduction rule must be applicable if the instance is too large compared to the parameter. Namely, we will use Lemma~\ref{lemma:partition} in the proofs of Lemmas~\ref{lemma:diameter},~\ref{lemma:n*A}, and~\ref{lemma:BBBleaf}.

\begin{definition}\label{def:additive-partition}
An \emph{$\alpha$-additive partition $\P$} of a graph $R$ is a partition $\P=\{V_i \mid i\in[|\P|]\}$ of $V(R)$ such that $\alpha(R) = \sum_{i \in [|\P|]}\alpha(R[V_i])$.
\end{definition}

Observe that if~$\P$ is an $\alpha$-additive partition of the graph~$R = G \setminus X$ of an instance~$(G,X,k)$, then the number of conflicts induced on~$R$ by~$X'$ is at least as large as the number of blocks~$V' \in \P$ of the partition for which~$\conf_{V'}(X') \neq 0$.

\begin{lemma}\label{lemma:partition}
Let $(G,X,k)$ be an input of \pb, $R'$ be a connected component of $R$, and $\P$ be an $\alpha$-additive partition of $G[R']$.
Suppose $\P = \P_1 \cup \P_2$, where each $V' \in \P_1$ is not $x$-almost-free.
Then, $|\P_1| < |\X| \cdot x$.
\end{lemma}
\begin{proof}
For any $V' \in \P_1$, as $V'$ is not $x$-almost-free there exists $X^{V'} \in \X$ such that $\conf_{V'}(X^{V'}) \neq 0$ and $\conf_R(X^{V'}) < x$.
Suppose for a  contradiction that $|\P_1| \ge |\X| \cdot x$. By the pigeonhole principle, this implies that there exists $\P'_1 \subseteq \P_1$ with $|\P'_1| \ge x$ and $X' \in \X$ such that $\conf_{V'}(X') \neq 0$ for all $V' \in \P'_1$, while $\conf_R(X') < x$. As $\P$ is an $\alpha$-additive partition, this implies that $\conf_{R'}(X') \ge |\P'_1| \ge x$.
Finally, as $\conf_{R}(X') \ge \conf_{R'}(X')$ since~$R'$ is a connected component of~$R$, we get our contradiction.
\end{proof}

Finally, we need the following technical lemma. It will be used in the safeness proofs of our rules to say that
if we have a chunk $X' \in \X$ that induces many conflicts on~$R$ in the instance $(G,X,k)$ before the reduction, then~$X'$ also induces many conflicts in the instance $(G',X,k)$ obtained after the reduction.
\begin{lemma}\label{lemma:shift}
Let $(G,X,k)$ and $(G',X,k)$ be two inputs of \pb such that $G[X]=G'[X]$ (implying that $G$ and $G'$ have the same set of chunks $\X$).
Suppose that there exist two non-negative integers $\Delta_1, \Delta_2$ such that
\begin{enumerate}
\item $\alpha(G' \setminus X) \ge \alpha(G \setminus X)-\Delta_1$ and\label{cond:612:first}
\item for any $X' \in \X$ we have $\alpha((G \setminus X) \setminus N_G(X')) \ge \alpha((G' \setminus X) \setminus N_{G'}(X'))-\Delta_2$.\label{cond:612:second}
\end{enumerate}
Then, for any $X' \in \X$ it holds that $\conf^{G'}_{V(G') \setminus X}(X') \ge \conf^G_{V(G) \setminus X}(X')-\Delta_1-\Delta_2$.
\end{lemma}
\begin{proof}
Let $X' \in \X$. We have
\begin{align*}
  \alpha((G' \setminus X') \setminus N_{G'}(X')) & \leq \alpha((G \setminus X) \setminus N_G(X'))+\Delta_2  & \mbox{By \eqref{cond:612:second}} \\
  & = \alpha(G \setminus X)-\conf^G_{V(G) \setminus X}(X')+\Delta_2 & \mbox{Definition \ref{def:conflict}} \\ 
  & \le \alpha(G' \setminus X)+\Delta_1-\conf^G_{V(G) \setminus X}(X')+\Delta_2. & \mbox{By \eqref{cond:612:first}}
\end{align*}
The desired bound follows.
\end{proof}

\subsubsection{Bounding the diameter}\label{subsubsec:diam}

In the following we say that a rule is a \emph{meta-rule} if it takes as input, in addition to $(G,X,k)$, a tree of bridges that will be computed in the course of the kernelization algorithm (see Algorithm~\ref{algo:kernel}).
The two following meta-rules are designed to shrink the diameter of the lowering tree of a connected component of~$R$.
Whenever we refer to a pending component~$H_{v}$ in the statement of a meta-rule, these should be understood to be pending components for the graph~$R = G \setminus X$, so that a pending component for a tree~$T$ of bridges in~$R$ is a connected component of~$R \setminus E(T)$.

\begin{MetaruleN}\label{rule:*A}~\\
\textbf{Input:} An input $(G,X,k)$ of \pb and a tree of bridges $T$ of a connected component $R'$ of~$R$.\\
\textbf{Action: } If there exists a $T$-conflict structure of type 1, namely $C=H_{v_1} \cup H_{v_2}$ using the notation of Definition~\ref{def:conflictstrucutre}, such that $C$ is $(|X|+2)$-almost-free,
then remove edge $\{v_1,v_2\}$ and define $k'=k$.
\end{MetaruleN}

\begin{ruleproof}
Let $G'$ be the graph obtained by removing $\{v_1,v_2\}$. Let us only prove that $\alpha(G') \ge k' $ implies that $\alpha(G) \ge k$ as the other direction is straightforward.
Let $S'$ be an independent set of $G'$ with $|S'| \ge k$. If~$S'$ contains at most one endpoint of the removed edge~$\{v_1, v_2\}$, then it is an independent set in~$G$ and we are done. Suppose $\{v_1,v_2\} \subseteq S'$. We distinguish two cases.

\emph{Case 1}: $\conf_C(S'_X)=0$. In this case there exists an independent set $Z$ of $C$ such that $S'_X \cup Z$ is still an independent set in $G$, and
$|Z|=\alpha(G[C])$. Now observe that since at least one of~$v_1, v_2$ is of type~$A$, we have~$\alpha(G[C])=\alpha(G[H_{v_1}])+\alpha(G[H_{v_2}])$: if~$i^* \in [2]$ such that~$v_{i^*}$ is of type~$A$, then a maximum independent set of~$G[H_{v_{i^*}}]$ exists that does not use~$v_{i^*}$, and it combines with a maximum independent set of~$G[H_{v_{3-i^*}}]$ into a maximum independent set of~$G[C]$. This implies that $|Z| \ge |S'_C|$.
For each~$i \in [2]$, since the pending components arise by removing the edges of a tree of bridges from~$R = G \setminus X$, it follows that~$v_i$ is the only vertex of~$H_{v_i}$ that potentially has neighbors in~$R$ outside of~$H_{v_i}$.
Since $Z \cap \{v_1,v_2\} \subseteq S'_C \cap \{v_1,v_2\}$, we get that $Z \cup S'_{R' \setminus C}$ is also an independent set in $G$.
Thus, $Z \cup S'_X \cup S'_{R' \setminus C} \cup S'_{R \setminus R'}$ is an independent set in $G$ of size at least $k$.

\emph{Case 2}: $\conf_C(S'_X) \neq 0$. According to Lemma~\ref{lemma:smallkiller}, there exists $X' \in \X$ with $X' \subseteq S'_X$ such that $\conf_C(X') \neq 0$.
As $C$ is $(|X|+2)$-almost-free, this implies that $\conf^G_R(X') \ge |X|+2$.
Observe that $\alpha(G' \setminus X) \ge \alpha(G \setminus X)$, and that for any $X' \in \X$ we have $\alpha((G \setminus X) \setminus N_G(X')) \ge \alpha((G' \setminus X) \setminus N_G(X'))-1$.
Thus, we can apply Lemma~\ref{lemma:shift} with $\Delta_1=0$ and $\Delta_2=1$, and we get that $\conf^{G'}_{V(G') \setminus X}(X') \ge |X|+1$.
This implies that $|S'_{V(G') \setminus X}| \le \alpha(G' \setminus X)-(|X|+1) \le \alpha(G \setminus X) -|X|$, implying in turn that $\alpha(G[R]) \ge |S'|$.
\end{ruleproof}


\begin{MetaruleN}\label{rule:B2B2}~\\
\textbf{Input: } An input $(G,X,k)$ of \pb and  a tree of bridges $T$ of a connected component $R'$ of $R$.\\
\textbf{Action: } If there exists a $T$-conflict structure of type 2, namely $C=H_{v_1} \cup H_{v_2}$, with a path $(u_2,v_1,v_2,u_1)$ in $T$  using the notation of Definition~\ref{def:conflictstrucutre}, such that $C$ is $(|X|+1)$-almost-free, then identify $u_i$ and $v_i$ for each~$i \in [2]$, and denote by $w_i$ the obtained vertex. Finally, define $k'=k-1$.
\end{MetaruleN}
\begin{ruleproof}
  Let $G'$ be the graph obtained after applying the rule and
  let $U = \{u_1,u_2\}$.

  Let us first prove  that $\alpha(G) \ge k$ implies that $\alpha(G') \ge k'$.
  Let $S$ be an independent set of $G$ with $|S| \ge k$. Recall that~$S_X = S \cap X$. We distinguish several cases.

  \emph{Case 1}: $|S \cap U| = 0$. In this case, $S'=S \setminus \{v_1,v_2\}$ has size at least $k-1$ and is still
  an independent set in $G'$.

  \emph{Case 2}: $|S \cap U| = 1$. Let $\{u_i\} = S \cap U$. If $S$ also contains $v_i$ then define $S'=(S \setminus \{u_i,v_i\}) \cup \{w_i\}$.
  Then~$S'$ is still an independent set in $G'$ as, in particular, $N_{G'}(w_i) \cap S' = (N_G(u_i) \cup N_G(v_i)) \cap S = \emptyset$, as $S$ contains $\{u_i,v_i\}$.
  Otherwise ($S$ does not contain $v_i$), define $S' = S \setminus \{u_i\}$.
  \emph{Case 3}: $|S \cap U| = 2$. We distinguish two subcases.

  \emph{Case 3.1}:  $\conf_C(S_X)=0$. This implies that there exists an independent set $Z \subseteq C$ such that $|Z|=\alpha(G[C])$ and such that $S_X \cup Z$ is still an independent set in $G$.
  As $v_1$ and $v_2$ have type $B$, we get that $\alpha(G[C])=\alpha(G[H_{v_1}])+\alpha(G[H_{v_2}])-1$. Moreover, we have $|Z \cap \{v_1,v_2\}| = 1$.
  Suppose $Z \cap \{v_1,v_2\} = \{v_1\}$; the other case follows symmetrically.
	For~$i \in [2]$, recall that~$v_i$ is the only vertex of~$H_{v_i}$ that potentially has neighbors in~$R$ outside of~$H_{v_i}$.
	Hence we get that $Z \cup Z'$ is also an independent set in $G$ where $Z' = S_{R' \setminus C} \setminus \{u_2\}$.
  This implies that $\tilde{S}=Z \cup Z' \cup S_X \cup S_{R \setminus R'}$ is an independent set of $G$.
  Moreover, we get that $|\tilde{S}| \ge k$. Indeed, as $S$ contains $u_1$ and $u_2$, it does not contain $v_1$ nor $v_2$, implying that $|S_C| < \alpha(G[C])$ as
  both $v_i$ have type $B$. This implies that $|Z \cup Z'| \ge |S_C \cup S_{R'\setminus C}|$, leading to $|\tilde{S}| \ge k$.
  Finally, we define $S' = (\tilde{S}  \setminus \{u_1,v_1\}) \cup \{w_1\}$ and we get that $S'$ is an independent set in $G'$ of size at least $k-1$.

  \emph{Case 3.2}: $\conf_C(S_X)\neq 0$. According to Lemma~\ref{lemma:smallkiller}, there exists $X' \in \X$ with $X' \subseteq S_X$ such that $\conf_C(X') \neq 0$.
  As $C$ is $(|X|+1)$-almost-free, this implies that $\conf^G_R(X') \ge |X|+1$. This in turn implies that $\alpha(G[R]) \ge |S|+1$. As $\alpha(G[R']) \ge \alpha(G[R])-1$, we get the desired result.

\medskip

  Let us now prove that $\alpha(G') \ge k'$ implies that $\alpha(G) \ge k$.
  Let $W = \{w_1,w_2\}$.
  Let $S'$ be an independent set of $G'$ with $|S'| \ge k-1$.
  Suppose first that $|S' \cap W|=1$ and let~$w_i \in S' \cap W$. As $w_i$ is the result of identifying the nonadjacent vertices $u_i$ and $v_i$,
  we get that $(S' \setminus \{w_i\}) \cup \{u_i,v_i\}$ is an independent set of $G$ of size at least $k$.
  Since the identification has made~$w_1$ and~$w_2$ adjacent in~$G'$, in the remainder we have $|S' \cap W|=0$. Let $H'_{v_i}=H_{v_i} \setminus \{v_i\}$ for $i \in [2]$, and let $V'$ such that we get a partition $R' = C \cup \{u_1,u_2\} \cup V'$ that satisfies~$N_G(V') \cap C = \emptyset$.
  Let us partition $S' = S'_X \cup S'_{H'_{v_1}} \cup S'_{H'_{v_2}} \cup S'_{V'} \cup S'_{R \setminus R'}$. As $|S' \cap W|=0$ and as both $v_i$ are of type $B$,
  we get that $|S'_{H'_{v_1}} \cup S'_{H'_{v_2}}| \le \alpha(G[C])-1$. We distinguish two cases.

  \emph{Case 1}: $\conf_C(S'_X)=0$. This implies that there exists an independent set $Z$ of $G[C]$ such that $|Z|=\alpha(G[C])$ and
 $Z \cup S'_X$ is an independent set in $G$. Moreover, $N_G(V') \cap C = \emptyset$ implies that $Z \cup S'_{V'}$ is an independent set in $G$.
 This implies that $S'_X \cup Z \cup S'_{V'} \cup S'_{R \setminus R'}$ is an independent set in $G$ of size at least $k$.

  \emph{Case 2}: $\conf_C(S'_X)\neq 0$. According to Lemma~\ref{lemma:smallkiller}, there exists $X' \in \X$ with $X' \subseteq S'_X$ such that $\conf_C(X') \neq 0$.
As $C$ is $(|X|+1)$-almost-free, this implies that $\conf^G_R(X') \ge |X|+1$.
Observe that $\alpha(G' \setminus X) \ge \alpha(G \setminus X)-1$, and that for any $X' \in \X$, $\alpha((G \setminus X) \setminus N_G(X')) \ge \alpha((G' \setminus X) \setminus N_G(X'))$.
Thus, we can apply Lemma~\ref{lemma:shift} with $\Delta_1=1$ and $\Delta_2=0$, and we get that $\conf^{G'}_{V(G') \setminus X}(X') \ge |X|$.
This implies that $|S'_{V(G') \setminus X}| \le \alpha(G' \setminus X)-|X| \le \alpha(G \setminus X)-|X|$, implying in turn that $\alpha(G[R]) \ge |S'|$.
\end{ruleproof}

Notice that in Algorithm~\ref{algo:kernel} between Lines~\ref{algo:T} and~\ref{algo:diambounded}, we first compute a lowering tree $T$, but we try to apply
Meta-Rules~\ref{rule:*A} and~\ref{rule:B2B2} on a longest path $T'$ of $T$ (and not directly on $T$). Observe also, as depicted in Figure~\ref{fig:diam}, that types of vertices may change when
considered in $T$ or $T'$.
Let us now prove that if neither Meta-Rule~\ref{rule:*A} nor Meta-Rule~\ref{rule:B2B2} can be applied to a longest path $T'$ of a lowering tree $T$, then $T'$ has bounded length.

\begin{figure}[t]
\begin{center}
\includegraphics[width=\textwidth]{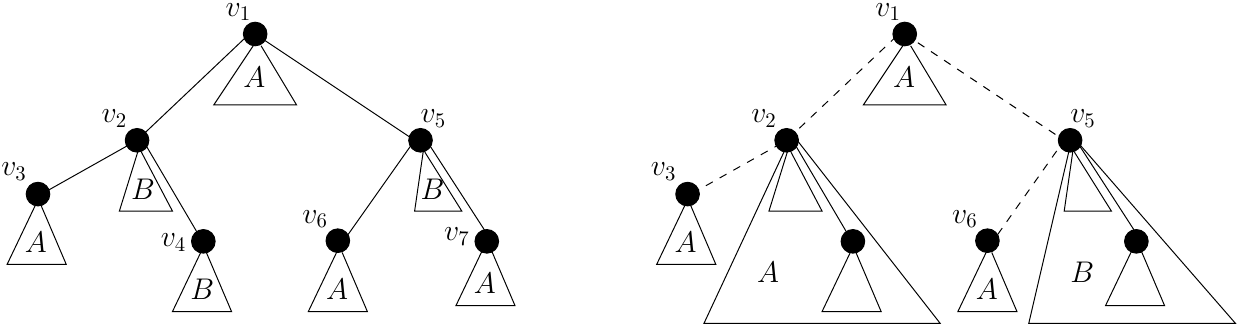}
\end{center}
\caption{On the left: example of a lowering tree $T$ (with $V(T)=\{v_i \mid i \in [8]\}$) computed in Line~\ref{algo:T} of Algorithm~\ref{algo:kernel}. On the right: a longest path $T'=(v_3,v_2,v_1,v_5,v_6)$ computed in Line~\ref{algo:T'}.
The triangle below each vertex represents its pending component, and the type of each vertex is indicated inside its triangle.
Notice that types of vertices in $T'$ may be different from types in $T$. For example, $H^{T'}_{v_2} = H^T_{v_2} \cup H^T_{v_4}$, implying that $v_2$ has type $A$ in $T'$, whereas
it has type $B$ in $T$.}
\label{fig:diam}
\end{figure}

\begin{lemma}\label{lemma:diameter}
Let $(G,X,k)$ be an input of \pb, $R'$ be a connected component of $R$, $T$ be a lowering tree of $R'$, and $T'$ be a longest path of $T$.
Suppose that neither Meta-Rule~\ref{rule:*A}  nor Meta-Rule~\ref{rule:B2B2} can be applied to $((G,X,k),T')$.
Then, $|V(T')| \le \O(|\X|\cdot |X|)$.
This implies that after Line~\ref{algo:diambounded} of Algorithm~\ref{algo:kernel}, $\diam(T) = \O(|\X|\cdot|X|)$.
\end{lemma}
\begin{proof}
Let $T' = (v_1,\ldots,v_t)$.
We consider types and pending components in $T'$.
Let us define a \emph{block} as an inclusion-maximal subpath of $T'$ with only type $B$ vertices.
Our goal is to define an $\alpha$-additive partition of $R'$.
For any block $L=(v_f,\ldots,v_\ell)$, let $M^L$ be the maximum matching using all vertices of $L$ if $|L|$ is even, and all vertices except $v_{\ell}$ otherwise.
Let $M_1 = \bigcup_{\text{$L$ is a block}}M^L$. For a block $L=(v_f,\dots,v_\ell)$ such that $v_{\ell} \neq v_t$ and $L$ has an odd number of vertices, let $e^L = \{v_{\ell},v_{\ell+1}\}$,
where $v_{\ell+1}$ has type $A$. Let $M_2$ be the union of all such $e^L$ edges.
Let $M_3$ be a maximum matching of type $A$ vertices not belonging to $M_1 \cup M_2$.
Finally, let $N = V(T') \setminus (V(M_1) \cup V(M_2) \cup V(M_3))$. Notice that vertices of $N$ are either type $A$ vertices between two blocks that were not matched in
$M_2$, or $v_t$ if the last block contains $v_t$ and has an odd number of vertices.

For any $e \in M_1 \cup M_2 \cup M_3$, $e=\{v_i,v_{i+1}\}$, let $C_e = H_{v_{i}} \cup H_{v_{i+1}}$.
Let $\P_1 = \{C_e \mid e \in M_1 \cup M_2 \cup M_3 \mid C_e \mbox{ is a $T'$-conflict structure (of type 1 or 2)}\}$.
Observe that $|\P_1| \le |M_1 \cup M_2 \cup M_3|-2$, where we subtract two because
if $\{v_1,v_2\}$ or $\{v_{t-1},v_t\}$ are in $M_1$, then the corresponding set $C_e$ is not a conflict structure of type 2, as we require that both vertices should be of degree two in $T'$.
Let $\P'_1 = \{C_e \mid e \in M_1 \cup M_2 \cup M_3\} \setminus \P_1$ and $\P_2 = \{H_v \mid v \in N\}$. Observe that $\P = \P'_1 \cup \P_1 \cup \P_2$ is a partition of $R'$.
Let us even prove that $\P$ is an $\alpha$-additive partition by constructing a maximum independent set of $R'$ of size $\sum_{V' \in \P} \alpha(G[V'])$.
For any $C_e, e \in M_1$, $e=\{v_i,v_{i+1}\}$, both $v_i$ and $v_{i+1}$ have type $B$, implying that $\alpha(G[C_e])=\alpha(H_{v_i})+\alpha(H_{v_{i+1}})-1$,
and thus we consider a maximum independent set $S_e$ of $G[C_e]$ that uses $v_i$ and not $v_{i+1}$.
For any $C_e, e \in M_2$, $e=\{v_i,v_{i+1}\}$, $v_i$ has type $B$ and $v_{i+1}$ has type $A$,
implying that $\alpha(G[C_e])=\alpha(H_{v_i})+\alpha(H_{v_{i+1}})$, and thus  we consider a maximum independent set $S_e$ of $G[C_e]$ that uses $v_i$ and not $v_{i+1}$.
For any $C_e, e \in M_3$, we consider a maximum independent set $S_e$ of $G[C_e]$ such that $S_e \cap V(T') = \emptyset$.
For $H_v, v \in N$, we consider a maximum independent set $S_v$ of $G[H_v]$ such that $S_v \cap V(T') = \emptyset$ if $v \neq v_t$, and $S_v \cap V(T') = \{v_t\}$ otherwise.
Observe that by definition $\bigcup_{e \in M_1 \cup M_2 \cup M_3}S_e \cup \bigcup_{v \in N}S_v$ is an independent set of $R'$ of the claimed size.

Now, as by hypothesis neither Meta-Rule~\ref{rule:*A}  nor Meta-Rule~\ref{rule:B2B2} can be applied to $((G,X,k),T')$, none of the $T'$-conflict structures of type $1$ of $\P_1$ is ($|X|+2$)-almost-free, and none of the $T'$-conflict structures of type $2$ of $\P_1$ is ($|X|+1)$-almost-free, thus not  ($|X|+2$))-almost-free either.
Thus, for any $C_e \in \P_1$, $C_e$ is not $(|X|+2)$-almost-free.
By Lemma~\ref{lemma:partition}, this implies that $|\P_1| \le |\X|  \cdot (|X|+2)$. As $|\P|=\O(|\P_1|)$ and $|V(T')|=\O(|\P|)$, we get the desired result.
\end{proof}

\subsubsection{Bounding the degree and the number of leaves}\label{subsubsec:degleaves}

We define another meta-rule based on the two previous ones.
Informally, in this meta-rule we move $v$ to the modulator, apply rules to decrease the number of connected components, and move $v$ back to $R$.

\begin{MetaruleN}\label{rule:deg}~\\
\textbf{Input: }  An input $(G,X,k)$ of \pb and  a tree  of bridges $T$ of a connected component $R'$ of $R$.\\
\textbf{Action: } If there exists $v \in V(T)$ such that $d_T(v) > 3|\X| \cdot |X|$ then
\begin{itemize}
\item apply exhaustively Rule~\ref{rule:free} and Rule~\ref{rule:chunkdegree} to $(G, X \cup \{v\},k)$, getting an instance $(G', X \cup \{v\},k')$, and
\item define the new instance as $(G',X,k')$.
\end{itemize}
\end{MetaruleN}
\begin{lemma}
  The following hold:
  \begin{itemize}
  \item Meta-Rule~\ref{rule:deg} is safe,
  \item $X$ is still a $c$-\bd-modulator in $G'$ (implying that $(G',X,k')$ is indeed an instance of \pb), and
  \item $|V(G')| < |V(G)|$.
  \end{itemize}
\end{lemma}
\begin{proof}
The safeness of this meta-rule follows directly from the safeness of Rule~\ref{rule:free} and Rule~\ref{rule:chunkdegree}.

Let us now prove the two other claims.
Recall that, by definition of $R$, we have the partition $V(G) = R \cup X$.
Observe that as $v \in V(T)$ and $T$ is a tree of bridges, $\cc(G[R'-\{v\}]) = d_T(v)$, implying by hypothesis that $\cc(G[R-\{v\}]) \ge \cc(G[R'-\{v\}]) > 3|\X|\cdot |X|$.
Notice also that in $(G,X \cup \{v\},k)$, $X \cup \{v\}$ is still a $c$-\bd-modulator.
Moreover, the set of chunks in this instance is $\X' = \{X' \subseteq  X \cup \{v\} \mid |X'| \le 2^c\}$, 
implying that $|\X'| \leq 2 |\X|$.

Choose $\tilde{R}$ such that we have the partition $V(G') = \tilde{R} \cup (X \cup \{v\})$.
After applying exhaustively Rule~\ref{rule:free} and Rule~\ref{rule:chunkdegree} to $(G,X \cup \{v\},k)$,
by Lemma~\ref{lemma:nbcc} it follows that $\cc(G'[\tilde{R}]) \le |\X'|\cdot|X \cup \{v\}|$.
As  $|\X'|\cdot|X \cup \{v\}| \le 2|\X| \cdot (|X|+1) \le 3|\X| \cdot |X|$, where we have assumed that $|X| \ge 2$ (as otherwise the problem can be solved in polynomial time), we deduce that at least one connected component of $G[R-\{v\}]$ has been deleted in order to obtain $G'$,
and thus that $|V(G')| < |V(G)|$.
Moreover, by Lemma~\ref{lemma:nbcc}
it follows that $G'$ is a subgraph of $G$, and that $\tilde{R}$ is a subset of $R-\{v\}$.
This implies that $G'[\tilde{R} \cup \{v\}]$ is a subgraph of $G[R]$, and thus that $X$ is still a $c$-\bd-modulator in $G'$.
\end{proof}

\begin{observation}\label{obs:degree}
If Meta-Rule~\ref{rule:deg} cannot be applied to $((G,X,k),T)$, then $\Delta(T) \le 3|\X|\cdot|X|$.
\end{observation}

Let us now bound the number of leaves of $T$. The following simple lemma will be used in the counting argument in the proof of Lemma~\ref{lemma:Tbounded}.
\begin{lemma}\label{lemma:absorb}
Let  $(G,X,k)$ be an input of \pb, let $T$ be  a tree of bridges   of a connected component $R'$ of $R$, and let $T_1$ be the tree obtained from $T$ by removing any leaf $v$ of type $A$ whose parent has type $B$.
Then $T_1$ has no $A$-leaf with $B$-parent, $\Delta(T_1) \le \Delta(T)$, and $\diam(T_1) \le \diam(T)$.
\end{lemma}

\begin{proof}
Let $v$ be a leaf of $T$ of type $A$ whose parent $u$ has type $B$ in $T$; we use the shortcuts $A$-leaf and $B$-parent.
Observe that $u$ still has type $B$ in $T_1$. Indeed, as $H^{T_1}_u = H^T_u \cup H^T_v$, we get that any maximum independent set of $G[H^{T_1}_u]$ necessarily contains $u$.
Thus, even if removing $v$ creates a new leaf $u$ in $T_1$, $u$ cannot be a $A$-leaf with $B$-parent in $T_1$, implying that $T_1$ has no $A$-leaf with $B$-parent.
This also implies the claimed bounds on $\Delta(T_1)$ and $\diam(T_1)$.
\end{proof}

In order to bound the number of leaves of $T$, we distinguish several types. Since our goal is to bound the size of $T$ by a polynomial in $|X|$, we may assume that $T$ contains at least three vertices, and therefore there exists a non-leaf vertex. In the following definition, we assume that the tree $T$ is rooted at an arbitrary non-leaf vertex, so that we can speak about the parent of a leaf; the definition is easily seen to be oblivious to the choice of the root, since we only need to consider parents of leaves.

\begin{definition}\label{def:cases}
Given an input $(G,X,k)$ of \pb and  a tree of bridges $T$ of a connected component $R'$ of $R$, let
\begin{itemize}
\item $n_{{\sf *A}}(T)  = |\{$$v \mid$ $v$ is a leaf of $T$ whose parent has type $A \}|$,
\item $n_{{\sf AB}}(T)  =  |\{$$v \mid$ $v$ is a leaf of $T$ of type $A$ whose parent has type $B \}|$,
\item $n_{{\sf BB2}}(T) = |\{$$v \mid$ $v$ is a leaf of $T$ of type $B$ whose parent $u$ has type $B$ and is such that $d_T(u)=2 \}|$,
\item $n_{{\sf BBB}}(T) = |\{$$v \mid$ $v$ is a leaf of $T$ of type $B$ whose parent $u$ has type $B$ and is such that $d_T(u)>2$ and there exists a leaf $v' \neq v$ of type $B$ with the same parent $u \}|$, and
\item $n_{{\sf BBbad}}(T) = |\{$$v \mid$ $v$ is a leaf of $T$ of type $B$ whose parent $u$ has type $B$ and  is such that $d_T(u)>2$ and $u$ is adjacent to a unique $B$-leaf$\}|$.
\end{itemize}
\end{definition}

Note that, since we may assume that $T$ is rooted at a non-leaf vertex, the cases in Definition~\ref{def:cases} are exhaustive, and each leaf of $T$ is counted in exactly one set.

\begin{lemma}\label{lemma:n*A}
Let $(G,X,k)$ be an input of \pb, $R'$ be a connected component of $R$, and $T$ be a lowering tree of $R'$.
Suppose that Meta-Rule~\ref{rule:*A}  cannot be applied to $((G,X,k),T)$.
Then, $n_{{\sf *A}}(T) \le \Delta(T) \cdot  |\X| \cdot (|X|+2)$.
This implies that after Line~\ref{algo:n*abounded} of Algorithm~\ref{algo:kernel}, $n_{{\sf *A}}(T_1) = \O((|\X|\cdot |X|)^2)$.
\end{lemma}
\begin{proof}
Let us define an $\alpha$-additive partition of $R'$.
Let $N_A$ be the set of vertices $v \in V(T)$ such that $v$ has type $A$ in $T$ and has at least one leaf adjacent to it.
For each $v \in N_A$, let $Y_v$ be the set of leaves of $T$ adjacent to $v$.
Define $C_v = \{H_u \mid u \in Y_v\} \cup H_v$. Let $\P_1 = \{C_v \mid v \in N_A\}$
and $\P_2 = \{V_2\}$, where $V_2 = V(R') \setminus \bigcup_{V' \in \P_1}V'$. We claim that $\P =\P_1 \cup \P_2$ is an $\alpha$-additive partition, as in
each $C_v$ we can take a maximum independent set $S_v$ of $G[C_v]$ such that $S_v$ does not contain $v$. Thus, defining $S'$ as a maximum independent set of $G[V_2]$,
we get that $(\bigcup_{v \in N_A}S_v) \cup S'$ is a maximum independent set of $R'$.

Let us now prove that for any $v \in N_A$, $C_v$ is not $(|X|+2)$-almost-free.
Let $v \in N_A$ and let $u \in Y_v$. As Meta-Rule~\ref{rule:*A} cannot be applied to $C=H_u \cup H_v$, we know that $C$ is not $(|X|+2)$-almost-free, implying that
there exists $X' \in \X$ such that $\conf_C(X') \neq 0$ and $\conf_R(X') < |X|+2$.
However, as $ \{C, C_v \setminus C \}$ is an $\alpha$-additive partition of $C_v$, $\conf_C(X') \neq 0$ implies $\conf_{C_v}(X') \neq 0$, and
thus that $C_v$ is not $(|X|+2)$-almost-free.
By Lemma~\ref{lemma:partition}, this implies that $|\P_1| \le |\X| \cdot (|X|+2)$. As $n_{{\sf *A}}(T) \le \Delta(T) \cdot |\P_1|$, we get the desired result.
\end{proof}



We now define a meta-rule to bound $n_{{\sf BBB}}(T)$.

\begin{MetaruleN}\label{rule:Bleaf}~\\
\textbf{Input: }  An input $(G,X,k)$ of \pb and  a tree of bridges $T$ of a connected component $R'$ of $R$.\\
\textbf{Action: } If there exists a $T$-conflict structure $C$ of type $3$, namely $C=H_u$ where $u$ has type $B$ and is a leaf of $T$, such that $C$
is $(|X|+2)$-almost-free, then remove vertex~$u$ and its parent~$v$ in $T$ from the graph~$G$, and define $k'=k-1$.
 \end{MetaruleN}
\begin{ruleproof}
Let $G'$ be the graph obtained after applying the rule.
Let us first prove that $\alpha(G) \ge k$ implies that  $\alpha(G') \ge k'$.
Let $S$ be an independent set of $G$ with $|S| \ge k$.
Then $S' = S \setminus \{u,v\}$ is an independent set of $G'$ of size at least $k'$.

Let us now prove that $\alpha(G') \ge k'$ implies that  $\alpha(G) \ge k$.
Let $S'$ be an independent set of $G'$ with $|S'| \ge k'$.
If $\conf^G_{C}(S'_X) = 0$ then there exists a maximum independent set $Z$ of $C$ such that $S'_X \cup Z$ is an independent set in $G$.
Observe that replacing  $(S' \cap (C \setminus u))$ by $Z$ in $S'$ gives an independent set $S$ in $G$.
Moreover, as $u$ has type $B$ we know~$|(S' \setminus (C \setminus u))| \le \alpha(G[C])-1$, so that $|S| \ge |S'|+1$.
Otherwise, if $\conf^G_{C}(S'_X) \neq 0$, according to Lemma~\ref{lemma:smallkiller}, there exists $X' \in \X$ with $X' \subseteq S'_X$ such that $\conf^G_C(X') \neq 0$.
As $C$ is $(|X|+2)$-almost-free, this implies that $\conf^G_R(X') \ge |X|+2$.
 Observe that $\alpha(G' \setminus X) \ge \alpha(G \setminus X)-1$, and that for any $X' \in \X$ we have $\alpha((G \setminus X) \setminus N_G(X')) \ge \alpha((G' \setminus X) \setminus N_G(X'))$.
Thus, we can apply Lemma~\ref{lemma:shift} with $\Delta_1=1$ and $\Delta_2=0$, and we get that $\conf^{G'}_{V(G') \setminus X}(X') \ge |X|+1$.
This implies that $|S'_{V(G') \setminus X}| \le \alpha(G' \setminus X)-(|X|+1) \le \alpha(G \setminus X)-(|X|+1)$, implying in turn that we can take for $S$ any maximum independent set of $G[R] = G \setminus X$ as $\alpha(G[R]) \ge |S'|+1$.
\end{ruleproof}


\begin{lemma}\label{lemma:BBBleaf}
Let $(G,X,k)$ be an input of \pb, $R'$ be a connected component of $R$, and $T$ be a lowering tree of $R'$.
Suppose that Meta-Rule~\ref{rule:Bleaf} cannot be applied to $((G,X,k),T)$.
Then, $n_{{\sf BBB}}(T) \le \Delta(T)   \cdot |\X| \cdot (|X|+2)$.
This implies that after Line~\ref{algo:nBBBbounded} of Algorithm~\ref{algo:kernel}, $n_{{\sf BBB}}(T_1) = \O((|\X| \cdot |X|)^2)$.
\end{lemma}
\begin{proof}
Let us define an $\alpha$-additive partition of $R'$.
Let $N_B$ be the set of vertices $v \in V(T)$ such that $v$ has type $B$ in $T$ and  has at least two  leaves of type $B$ adjacent to it.
For each $v \in N_B$, let $Y_v$ be the set of $B$ leaves of $T$ adjacent to $v$.
Define $C_v = \{H_u \mid u \in Y_v\} \cup H_v$. Let $\P_1 = \{C_v \mid v \in N_B\}$
and $\P_2 = \{V_2\}$, where $V_2 = V(R') \setminus \bigcup_{V' \in \P_1}V'$. We claim that $\P =\P_1 \cup \P_2$ is an $\alpha$-additive partition as in
each $C_v$, any maximum independent set $S_v$ of $G[C_v]$ does not contain $v$ as there are at least two  leaves of type $B$ adjacent to $v$. Thus, defining $S'$ as  maximum independent set of $G[V_2]$,
we get that $(\bigcup_{v \in N_B}S_v) \cup S'$ is a maximum independent set of $R'$.

Let us now prove that for any $v \in N_B$, $C_v$ is not $(|X|+2)$-almost-free.
Let $v \in N_B$ and let $u \in Y_v$. As Meta-Rule~\ref{rule:Bleaf} cannot be applied on $C=\{H_u\}$, we know that $C$ is not $(|X|+2)$-almost-free, implying that
there exists $X' \in \X$ such that $\conf_C(X') \neq 0$ and $\conf_R(X') < |X|+2$.
However, as $\{C, C_v \setminus C\}$ is an $\alpha$-additive partition of $C_v$, $\conf_C(X') \neq 0$ implies $\conf_{C_v}(X') \neq 0$, and
thus that $C_v$ is not $(|X|+2)$-almost-free.
By Lemma~\ref{lemma:partition}, this implies that $|\P_1| \le |\X| \cdot (|X|+2)$. As $n_{{\sf BBB}}(T) \le \Delta(T) \cdot |\P_1|$, we get the desired result.
\end{proof}

Finally, our last meta-rule will be used to bound $n_{{\sf BB2}}(T)$.

\begin{MetaruleN}\label{rule:BB2leaf}~\\
\textbf{Input: }  An input $(G,X,k)$ of \pb  and  a tree of bridges $T$ of a connected component $R'$ of $R$.\\
\textbf{Action: } If there exists a $T$-conflict structure $C$ of type $4$, namely $C=H_{v_1} \cup H_{v_2}$ with a path $(v_1,v_2,u)$ in $T$ using notations of Definition~\ref{def:conflictstrucutre},
such that $C$ is $(|X|+1)$-almost-free, then identify $v_1$ and $u$, remove $v_2$, and define $k'=k-1$.
 \end{MetaruleN}

\begin{ruleproof}
Let $G'$ be the graph obtained after applying the meta-rule. We will prove the safeness of this meta-rule using the safeness of Meta-Rule~\ref{rule:B2B2}.
Let $G_1$ be the graph obtained from $G$ by adding two vertices $x$, $y$, and two edges $\{x,y\}$, $\{y,v_1\}$.
Let $R_1 = R' \cup \{x,y\}$.
It is immediate that $(G_1,X,k+1)$ is equivalent to $(G,X,k)$.
Now, as the two new edges are bridges in $G_1$, we get that $T_1 = T \cup \{\{x,y\}, \{y,v_1\}\}$ is a tree of bridges of $R_1$.
As $v_1$ and $v_2$ still have type $B$ in $T_1$, we get that $C$ is a  conflict structure of type $2$ in $T_1$, and is still $(|X|+1)$-almost-free.
Thus, we can apply Meta-Rule~\ref{rule:B2B2} to $(G_1,X,k+1)$ with $T_1$, and get an equivalent instance $(G_2,X,k)$.
Now, observe that $x$ still has degree one in $G_2$, and thus by defining $G_3 = G_2 \setminus (\{x\} \cup N_{G_2}(x))$, we get that $(G_3,X,k-1)$ is equivalent to $(G_2,X,k)$.
As $G_3 = G'$, we get the desired result.
\end{ruleproof}

\begin{lemma}\label{lemma:BB2leaf}
Let $(G,X,k)$ be an input of \pb, $R'$ be a connected component of $R$, and $T$ be a lowering tree of $R'$.
Suppose that Meta-Rule~\ref{rule:BB2leaf}  cannot be applied to $((G,X,k),T)$.
Then, $n_{{\sf BB2}}(T) \le \Delta(T) \cdot  |\X| \cdot (|X|+1) $.
This implies that after Line~\ref{algo:nBB2bounded} of Algorithm~\ref{algo:kernel}, $n_{{\sf BB2}}(T_1) = \O((|\X|\cdot|X|)^2)$.
\end{lemma}
\begin{proof}
The proof follows the same approach as those of Lemmas~\ref{lemma:n*A} and~\ref{lemma:BBBleaf}.
The corresponding $\alpha$-additive partition of $R'$ used here is $\P =\P_1 \cup \P_2$,
where $\P_1$ is the set of $T$-conflict structures of type $4$ and $\P_2 = \{V_2\}$, where $V_2 = V(R') \setminus \bigcup_{V' \in \P_1}V'$. (We use the observation here that if~$T$ has at least four vertices, then any two distinct $T$-conflict structures of type $4$ are vertex-disjoint, so that this indeed gives a valid partition. If~$T$ has at most three vertices, the lemma trivially holds.)
\end{proof}

Let us now prove that if no meta-rule can be applied to the subtree $T_1$ defined in Algorithm~\ref{algo:kernel}, then $|V(T)|$ is bounded.
\begin{lemma}\label{lemma:Tbounded}
Let $(G,X,k)$ be an input of \pb, $R'$ be a connected component of $R$, and $T$ be a lowering tree of $R'$.
Then, after Line~\ref{algo:nBB2bounded} of Algorithm~\ref{algo:kernel}, it holds that
$|V(T)| = \O(|\X|^5 \cdot |X|^5)$.
\end{lemma}
\begin{proof}
Let $T_1$ be the tree obtained from $T$ by removing any leaf $v$ of type $A$ whose parent has type $B$ (thus, $n_{{\sf AB}}(T_1)=0$ by Lemma~\ref{lemma:absorb}).
Let us root $T_1$ arbitrarily at a non-leaf vertex (since our intermediate goal is to bound the size of $T_1$ by a polynomial on $|\X|$ and $|X|$, we may assume that $|V(T_1)| \geq 3$, hence $T_1$ contains a non-leaf vertex).
Let $f$ be the number of leaves of $T_1$. It is sufficient to bound $f$, as $|V(T_1)| \le f \cdot \diam(T_1)$.
Let $P$ be the set of type $B$ vertices $v$ of $T_1$ such that $v$ is adjacent to exactly one $B$-leaf, and $d_{T_1}(v) > 2$. By the definition of $n_{{\sf BBbad}}(T_1)$, it holds that $|P| = n_{{\sf BBbad}}(T_1)$.
Observe that we get $f = n_{{\sf *A}}(T_1)+n_{{\sf AB}}(T_1)+n_{{\sf BBB}}(T_1)+n_{{\sf BB2}}(T_1)+n_{{\sf BBbad}}(T_1) = n_{{\sf *A}}(T_1)+n_{{\sf BBB}}(T_1)+n_{{\sf BB2}}(T_1)+|P|$, where we have used that $n_{{\sf AB}}(T_1) = 0$ and $ n_{{\sf BBbad}}(T_1) = |P|$.
Let us now bound $|P|$.

Let $P'$ be the set of vertices of $P$ which have no vertex of $P$ as a proper descendant.
Observe that $|P| \le |P'|\cdot \diam(T_1)$, and thus it now remains to bound $|P'|$.

Let $p \in P'$ and let $Y_p = \{u_1,\ldots,u_{x_p}\}$ be the children of $p$ in $T_1$. Observe that ${x_p} \ge 2$ as every vertex $v \in P$ satisfies $d_{T_1(v)}>2$.
As there is exactly one type $B$-leaf in $Y_p$, w.l.o.g. let us assume it is $u_1$.
Then, for any $i \ge 2$, let $T_p^i$ be the subtree of $T_1$ rooted at~$u_i$. 
(Note that $|V(T_p^i)| \ge 2$ since~$T_1$ has no leaf of type $A$ with parent of type~$B$.) 
By definition of $P'$, $T_p^i$ has no vertex of $P$. Therefore, for any two distinct vertices $p_1,p_2 \in P'$, and any two indices $i \in [2,x_{p_1}]$ and $j \in [2,x_{p_2}]$, it follows that $V(T^i_{p_1}) \cap V(T^j_{p_2})= \emptyset$.   
Thus, for every $p \in P'$ and every $i \in [2,x_p]$, every leaf of  $T_p^i$ contributes at least one to $n_{{\sf *A}}(T_1)+n_{{\sf BBB}}(T_1)+n_{{\sf BB2}}(T_1)$. Since $x_p \geq 2$ for every $p \in P'$, every $p \in P'$ contributes at least one to $n_{{\sf *A}}(T_1)+n_{{\sf BBB}}(T_1)+n_{{\sf BB2}}(T_1)$. This implies that $|P'| \le n_{{\sf *A}}(T_1)+n_{{\sf BBB}}(T_1)+n_{{\sf BB2}}(T_1)$. Thus, we have that
\begin{eqnarray*}
  |V(T_1)| & \leq & (n_{{\sf *A}}(T_1)+n_{{\sf BBB}}(T_1)+n_{{\sf BB2}}(T_1)+|P|) \cdot \diam(T_1) \\
   & \leq & (n_{{\sf *A}}(T_1)+n_{{\sf BBB}}(T_1)+n_{{\sf BB2}}(T_1)+|P'|  \cdot \diam(T_1)) \cdot \diam(T_1) \\
    & \leq & ((n_{{\sf *A}}(T_1)+n_{{\sf BBB}}(T_1)+n_{{\sf BB2}}(T_1)) \cdot (1 + \diam(T_1))) \cdot \diam(T_1) \\
   & \leq & 2((n_{{\sf *A}}(T_1)+n_{{\sf BBB}}(T_1)+n_{{\sf BB2}}(T_1)) \cdot  \diam(T_1)^2 \\
  & \leq & 2((n_{{\sf *A}}(T_1)+n_{{\sf BBB}}(T_1)+n_{{\sf BB2}}(T_1)) \cdot  \diam(T)^2,
\end{eqnarray*}
\noindent where the last inequality follows from Observation~\ref{obs:degree}. Applying Lemma~\ref{lemma:diameter}, Lemma~\ref{lemma:n*A},
Lemma~\ref{lemma:BBBleaf}, and
Lemma~\ref{lemma:BB2leaf} it follows that $|V(T_1)| = \O(|\X|^4 \cdot |X|^4)$. Finally, as $|V(T)| \le \Delta(T) \cdot |V(T_1)|$ and $\Delta(T) = \O(|\X| \cdot |X|)$ by Observation~\ref{obs:degree}, we obtain that $|V(T)| = \O(|\X|^5 \cdot |X|^5)$ and the lemma follows.
\end{proof}

\subsection{Applying recursion}
\label{sec:recursion}

As now the size of the lowering tree of each connected component of $R$, and the number of connected components of $R$, are bounded by a polynomial of $|X|$,
we can move all these lowering trees in the modulator and recurse to get the polynomial kernel for \pb.

\begin{theorem} \label{thm:kernel:is}
Algorithm~\ref{algo:kernel} is a polynomial kernel for \pb of size $\O_c(|X|^{2^{f(c)}})$, where the constant hidden in $\O_c$ depends on $c$, and $f(c)=\sum_{i=5}^{c+4}i=\O(c^2)$.
\end{theorem}
\begin{proof}
Observe first that Algorithm~\ref{algo:kernel} is indeed polynomial for any fixed $c$ as Lemma~\ref{lemma:conflictandcopoly} and Lemma~\ref{lemma:typesandcopoly} imply that
all rules and meta-rules can be applied in polynomial time. The fact that the output is an equivalent instance immediately follows from the safeness of all rules and meta-rules.
Let us now bound the size of the kernel by induction on $c$.
For $c=0$, $f(c)=0$ and Algorithm~\ref{algo:kernel} outputs $X$, which has the required size.
Let $c \ge 1$.
Let $X_2$ be defined as in Line~\ref{algo:end} of Algorithm~\ref{algo:kernel}.
Let $x = |X|$.
Observe first that $|\X| = \O_c(x^{2^c})$.
At Line~\ref{algo:end} we know that $\cc(R) \le |\X| \cdot x$, and that for each connected component $R'$ of $R$, its lowering tree $T$ has $|V(T)| = \O(|\X|^5 \cdot x^5)$.
This implies that at Line~\ref{algo:end} we have $|X_1| = \O(|\X|^6 \cdot x^6) = \O_c(x^{2^{c+3}+6}) = \O_c(x^{2^{c+4}})$. Thus, the new $(c-1)$-\bd-modulator $X_2$ has size
$|X_2| = x+|X_1| = \O_c(x^{2^{c+4}})$.
By induction hypothesis, Algorithm~\ref{algo:kernel} with input $(G,X_2,k,c-1)$ will return an equivalent instance of size $\O_c(|X_2|^{2^{f(c-1)}})$, implying the claimed bound.
\end{proof}

The degree-bound of the kernel size bound of Theorem~\ref{thm:kernel:is} is exponential in~$c$. This is known to be unavoidable. Indeed, this can be proven directly from the fact that graphs of bridge-depth~$c$ have minimal blocking sets of size~$2^c$ and are closed under disjoint union. Alternatively, an explicit lower bound is given in the literature~\cite[Thm.~2]{JansenP18} that shows that \textsc{Vertex Cover} parameterized by a tree-depth-$c$ modulator~$X$ does not admit a kernel with~$(|X|^{2^{c-4} - \varepsilon})$ bits for any~$\varepsilon > 0$, unless \containment. Since~$\bd(G) \leq \td(G)$ for every graph $G$, the same lower bound holds for bridge-depth.



Theorem~\ref{thm:kernel:is}, together with the equivalence between \textsc{Independent Set} and \textsc{Vertex Cover} for structural parameterizations discussed in Section~\ref{sec:preliminaries}, proves the first implication of Theorem~\ref{thm:characterization}.
Indeed, given a graph $G \in \F$ with bounded bridge-depth, Algorithm~\ref{algo:kernel} will provide an equivalent instance $(G',X',k')$ of \pb of size $\O_c(|X|^{2^{f(c)}})$.
As $G'[V(G') \setminus X']$ may not belong to $\F$, we move the entire graph to the modulator and say that $(G',V(G'),k')$ is an equivalent instance with
$G'[V(G') \setminus X'] = \emptyset \in \F$.

\begin{theorem} \label{thm:kernel:F}
Let~$\mathcal{F}$ be a minor-closed family of graphs.
If~$\mathcal{F}$ has bounded bridge-depth, then \textsc{Vertex Cover} parameterized by vertex-deletion distance to~$\mathcal{F}$ has a polynomial kernelization.
\end{theorem}

\subsection{Negative result and proof of Theorem \ref{thm:characterization}}
\label{sec:negative-kernel}

The following negative result following easily from the results discussed so far, by using known gadgets.

\begin{theorem}\label{thm:nokernel}
Let $\F$ be a minor-closed family of graphs with unbounded bridge-depth (or equivalently, by Theorem~\ref{thm:bridgedepth:blockingsets}, with unbounded inclusion-minimal blocking sets).
Then, \IS/\distto$\F$ and \VC/\distto$\F$ do not admit a polynomial kernel unless \containment.
\end{theorem}
\begin{proof}
Let $t \ge 1$. An \emph{extended triangle path of length $t$} is a triangle-path of length $t$ (see Definition~\ref{def:triangle-path}) with two extra vertices $u$ and $v$, and two extra edges $\{u,a_1\}$ and $\{b_t,v\}$.
Let $\F^{{\sf etp}}$ be the family containing all graphs composed of disjoint unions of extended triangle paths.
It is known from~\cite[Thm.~2]{FominS16} that \IS/\distto$\F^{{\sf etp}}$ does not admit a polynomial kernel unless \containment, not even when a modulator is given together with the input.
(In fact, ~\cite[Thm.~2]{FominS16} proves that \IS parameterized by distance to mock forests does not admit a polynomial kernel unless \containment,
but their construction even produces graphs that are disjoint unions of extended triangle paths.)
The result for \IS is immediate as $\F^{{\sf tp}} \subseteq \F$ by Corollary~\ref{obs:Funbounded} (recall that $\F^{{\sf tp}}$ is the family of all triangle-paths), and as any disjoint union of extended triangle paths is a minor of a sufficiently large triangle-path. The same lower bound for \VC follows since there are parameter-preserving reductions in both ways.
\end{proof}

Theorem~\ref{thm:kernel:F} and Theorem~\ref{thm:nokernel} together imply Theorem~\ref{thm:characterization}.


\section{Conclusion} \label{sec:conclusion}
In this paper we introduced the graph parameter bridge-depth and used it to characterize the minor-closed graph classes~$\F$ for which \textsc{Vertex Cover} parameterized by~$\mathcal{F}$-modulator has a polynomial kernel. It would be interesting to see whether the characterization can be extended to subgraph-closed or even hereditary graph classes. If a characterization exists of the hereditary graph classes whose modulators lead to a polynomial kernel, it will likely not be as clean as Theorem~\ref{thm:characterization}: it will have to deal with the fact that bipartite graphs can be arbitrarily complex in terms of width parameters, while bipartite modulators allow for a polynomial kernel. Hence such a characterization has to capture parity conditions of~$\mathcal{F}$.

\medskip

A natural attempt to generalize our approach to deal with bipartite graphs is to consider the following parameter, which we call \emph{bipartite-contraction-depth}: we mimic the definition of bridge-depth (cf. Definition~\ref{def:bd}), except that we redefine the graph $\bar{G}$ to be the graph obtained from $G$ by \emph{simultaneously} contracting all edges that do \emph{not} lie on an odd cycle. Note that bipartite-contraction-depth generalizes bridge-depth, in the sense that bridges do not lie on an odd cycle, and that the bipartite-contraction-depth of a graph with an odd cycle transversal of size $k$ is at most $k+1$. Having defined this parameter, we would need, in order to obtain a statement similar to Theorem~\ref{thm:largebd:largembs}, that large bipartite-contraction-depth implies the existence of structures that allow to obtain kernel lower bounds, similarly to the fact that large bridge-depth implies the existence of large triangle-paths (cf. Corollary~\ref{obs:Funbounded}). The appropriate structure here seems to be an \emph{odd-cycle-path of length $t$}, defined as a set of $t$ vertex-disjoint odd cycles $C_1,\ldots,C_t$, and a set of $t-1$ vertex-disjoint paths (of any length) connecting $C_i$ to $C_{i+1}$ for $i\in [t-1]$, in such a way that for every $i \in \{2,\ldots,t-1\}$, the two attachment vertices in $C_i$ are distinct. Now the expected property would be that large bipartite-contraction-depth forces long odd-cycle-paths. Unfortunately, this is not true. Indeed, consider the \emph{Escher wall of size $h$} depicted in \cite[Fig. 3]{RautenbachR01}. It is proved in\cite{RautenbachR01} that this graph does not contain two vertex-disjoint odd cycles, but a smallest hitting set for odd cycles has size $h$. Since there are no two vertex-disjoint cycles, a longest odd-cycle-path has length one. On the other hand, it can be easily verified that an Escher wall of size $h$ has bipartite-contraction-depth $\Omega(h)$. Informally, this can be seen by noting that, initially, all edges lie on an odd cycle, hence a vertex removal is required, and that each such removal cascades in a constant number of contractions until all edges lie again on an odd cycle. Since a smallest hitting set for odd cycles of an Escher wall of size $h$ has size $h$, the claimed bound follows. Therefore, summarizing this discussion, if one aims at a result similar to Theorem~\ref{thm:characterization} that also applies to families $\F$ containing bipartite graphs, it seems that significant new ideas are required.

\medskip

Another open research direction consists of a further algorithmic exploration of the merits of bridge-depth. We expect that several polynomial-space fixed-parameter tractable algorithms that work for graphs of bounded tree-depth~\cite{ChenRRV18,PilipczukW18} can be extended to work with bridge-depth instead. Which other ways to enrich the recursive definition of tree-depth lead to novel algorithmic insights?

\medskip

As for kernelization purposes, it is plausible that bridge-depth also characterizes the existence of polynomial kernels for other problems other than \textsc{Vertex Cover},  parameterized by the vertex-deletion distance of the input graph to a minor-closed graph class. For instance, the \textsc{Feedback Vertex Set} problem and the generalizations considered by Jansen and Pieterse~\cite{JansenP18} seem to be good candidates.

\medskip

\noindent \textbf{Acknowledgement}. We would like to thank the anonymous reviewers for helpful comments that improved the presentation of the manuscript.

\bibliography{references}

\end{document}